%% file: ulcbeta.tex
\pdfoutput=1

\documentclass{amsart}

\usepackage[protrusion=true,expansion=true]{microtype}

\usepackage{amsmath}

\usepackage[style=authoryear,
 maxnames=2,
 maxbibnames=99 ,                
 uniquename=init ,
]{biblatex}
\bibliography{literature.bib}

\usepackage{ownstyle}

\newcommand{\sourceurl}{\url{http://math.unice.fr/laboratoire/logiciels}}

\author{Benedikt Ahrens}

\title[Modules over relative monads for syntax and semantics]{Modules over relative monads \\ for syntax and semantics}

\begin{document}

\begin{abstract}

\setlength{\parindent}{0pt}

We give an algebraic characterization of the syntax and semantics of a class of languages
with variable binding.

We introduce a notion of \emph{2--signature}: such a signature specifies not only the \emph{terms} of a language, but also
\emph{reduction rules} on those terms.
To any 2--signature $S$ we associate a category of ``models'' of $S$. This category has an initial object, which 
integrates the terms freely generated by $S$, 
and which is equipped with reductions according to the inequations given in $S$.
We call this initial object the \emph{language generated by $S$}.
Models of a 2--signature are built from \emph{relative monads} and \emph{modules} over such monads.
Through the use of monads, the models --- and in particular, the initial model --- 
come equipped with a \emph{substitution operation} that is compatible with reduction in a suitable sense.

The initiality theorem is formalized in the proof assistant \textsf{Coq}, yielding a machinery which, when fed 
with a 2--signature, provides the associated programming language with reduction relation and certified substitution.

\end{abstract}

\maketitle

\tableofcontents

\section{Introduction}

\input{introduction}

\input{monad}

\input{signature}

\input{prop_arities_formal}

\section{Conclusions \& Future Work}
We have presented an initiality result for abstract syntax which integrates semantics specified by reduction rules by means of preorders. 
It is based on relative monads and modules over such monads.
Reduction rules are specified by \emph{inequations}, whose definition is largely inspired by 
\citeauthor{journals/corr/abs-0704-2900}'s equations \parencite{journals/corr/abs-0704-2900}.
For any 2--signature $(S,A)$ with classic inequations, we construct the initial object in the 
category of representations of $(S,A)$.
The theorem is proved formally in the proof assistant \textsf{Coq}.

On another line of work \parencite{ahrens_ext_init} we have extended 
\citeauthor{ju_phd}'s initiality result \parencite[Chap.\ 6]{ju_phd} for simply--typed syntax
to allow for models over \emph{varying} object types.
In this way initiality accounts for translation between languages over different types.
Both lines of work, varying object types and the integration of operational semantics, can be combined: 
in \parencite{DBLP:conf/wollic/Ahrens12} we present an initiality result which allows for specification of reductions
as well as change of object types.
As an example, we consider the language PCF with its usual small--step semantics.
By equipping $\ULCB$ with a representation of PCF, we obtain
a translation of PCF to $\ULCB$ which is faithful with respect to semantics.
Our approach should be extended to more complex type systems featuring dependent types or polymorphism.
Another interesting feature to work on are \emph{conditional} reductions/rewritings.
A more fine--grained account of reduction might be given by considering \emph{graphs} instead
of preordered sets.

\subsubsection*{Acknowledgments}
We are grateful to Andr\'e Hirschowitz and Marco Maggesi for invaluable discussions.
We thank the anonymous referees for their careful reading and suggestions,
which led to several improvements in the presentation.
We finally thank Thorsten Altenkirch and Conor McBride for their editorial work.

\printbibliography

\end{document}

%% file: introduction.tex
We give an algebraic characterization, via a universal property, of the programming 
language generated by a \emph{signature}.

More precisely, we define a notion of \emph{2--signature} which allows the specification of
the \emph{syntax} of a programming language --- via a 1--signature, say, $S$ --- 
as well as its \emph{semantics} in form of reduction rules, specified through 
a set $A$ of inequations over $S$.
To any 1--signature $S$ we associate a category of \emph{models of $S$}. 
Given a 2--signature $(S,A)$, the inequations of $A$ give rise to a \emph{satisfaction}
predicate on the models of $S$, and thus specify a full subcategory of models of $S$
which satisfy the inequations of $A$.
We call this subcategory the \emph{category of models of $(S,A)$}.
Our main theorem states that this category has an initial object ---
the programming language associated to $(S,A)$ ---,
which integrates the terms generated by $S$, equipped with the reduction relation
generated by the inequations of $A$.

The theorem has been fully certified in the proof assistant \textsf{Coq} \parencite{coq}.
The \textsf{Coq} theory files as well as documentation are available online at
\begin{center}
 \sourceurl .
\end{center}

\subsection{Summary}

We define a notion of 2--signature in order to specify the terms and reduction rules
of functional programming languages.
Given any 2--signature, we characterize its associated programming language 
  as initial object in some category, thus
giving an algebraic definition of abstract syntax with semantics.
This characterization of syntax with reduction rules is given in two steps:
\begin{enumerate}
 \item 
At first \emph{pure syntax} is characterized as initial object in some category.
Here we use the term ``pure'' to express the fact that no semantic aspects such as reductions on 
terms are considered. This characterization is actually a consequence of a result by \textcite{DBLP:conf/wollic/HirschowitzM07}
 (cf.\ Sec.\ \ref{sec:intro_pure_syntax}).
 \item  Afterwards we consider \emph{inequations} specifying \emph{reduction rules}. Given a set of reduction rules
          for terms, we build up on the preceding result to give an algebraic characterization of 
          syntax \emph{with reduction}.
\end{enumerate}

\noindent
In summary, the merit of this work is to give an algebraic characterization of syntax \emph{with reduction rules},
building up on such a characterization for \emph{pure} syntax given by \textcite{DBLP:conf/wollic/HirschowitzM07}.

Our approach is based on relative monads
\parencite{DBLP:conf/fossacs/AltenkirchCU10}  
from the category $\Set$ of sets to the category $\PO$ of preorders and modules over such monads.  
Compared to traditional monads, relative monads allow for different categories as domain and codomain.

We now explain the above two points in more detail:

\subsubsection{Pure Syntax}\label{sec:intro_pure_syntax}

An \emph{arity} specifies the type of a term constructor: 
it is given by a list of natural numbers, the length of which specifies the number 
of arguments of the constructor. The list entries specify the number of variables that are bound in the 
corresponding argument.
A {1--signature} is a family of arities.
To any 1--signature $S$ we associate a category $\Rep^{\Delta}(S)$ of \emph{representations} --- or ``models'' --- of $S$,
where a model of $S$ is built from a relative monad on the functor $\Delta : \Set\to\PO$.

This category has an initial object (cf.\ Lem.\ \ref{lem:initial_wo_ineq}), which integrates the terms
freely generated by the signature. We call this object the \emph{(pure) syntax associated to $S$}. 
As mentioned above, we use the term ``pure'' to distinguish this initial object from 
the initial object associated to a \emph{2--signature}, 
which gives an analogous characterization of syntax \emph{with reduction rules} (cf.\ below).

Initiality for pure syntax is actually a consequence of a related initiality theorem proved by \textcite{DBLP:conf/wollic/HirschowitzM07}:
in that work, the authors associate, to any signature $S$, a category $\Rep(S)$ of models of $S$, where a model is built from 
a monad over $\Set$.
We connect the corresponding categories by exhibiting a pair of adjoint functors (cf.\ Lem.\ \ref{lem:adjunction_reps}) 
between our category $\Rep^{\Delta}(S)$ of representations of $S$ and that of \citeauthor{DBLP:conf/wollic/HirschowitzM07}, 
  
\begin{equation*}   
       \begin{xy}
        \xymatrix@C=4pc{
                  **[l] \Rep(S) \rtwocell<5>_{U_*}^{\Delta_*}{'\bot} & **[r] \Rep^{\Delta}(S)
}
       \end{xy} \enspace .
\end{equation*}
We thus obtain an initial object in our category $\Rep^{\Delta}(S)$ using the fact that left adjoints are co\-con\-tin\-uous:
the image under the functor 
$\Delta_* : \Rep(S) \to \Rep^{\Delta}(S)$ 
of the initial object in the category $\Rep(S)$ is initial in $\Rep^{\Delta}(S)$.

\subsubsection{Syntax with Reduction Rules}

Given a 1--signature $S$, an \emph{$S$--inequation} $E = (\alpha, \gamma)$ associates a pair $(\alpha^R, \gamma^R)$ of \emph{parallel} 
morphisms in a suitable category
to any
representation $R$ of $S$.
In a sense made precise later, we can ask whether 
  \[  \alpha^R \enspace \leq \enspace \gamma^R \enspace , \]
due to our use of relative monads into the category $\PO$ of preorders.
If this is the case, we say that $R$ \emph{satisfies} the inequation $E$.
A \emph{2--signature} is a pair $(S,A)$ consisting of a 1--signature $S$, which specifies the terms of a language, together with a set $A$ of 
\emph{$S$--inequations}, which specifies reduction rules on those terms.
Given a 2--signature $(S,A)$, we call \emph{representation of $(S,A)$} any representation of $S$ that
satisfies each inequation of $A$.
The \emph{category of representations of $(S,A)$} is defined to be the full subcategory of representations of $S$ 
whose objects are representations of $(S,A)$.

We would like to exhibit an initial object in the category of representations of $(S,A)$,
and thus must rule out inequations which are never satisfied.
We call \emph{classic} $S$--inequation any $S$--inequation whose codomain is of a particular form.
Our main result states that for any set $A$ of classic $S$--inequations the category of representations of $(S,A)$ has an initial object.
The class of classic inequations is large enough to account for the 
fundamental reduction rules; 
in particular, beta and eta reductions are given by classic inequations.

Our definitions ensure that any reduction rule between terms that is expressed by an inequation $E\in A$ is automatically 
propagated into subterms. 
The set $A$ of inequations hence only needs to contain some ``generating'' inequations, 
a fact that is well illustrated by our example 2--signature $\Lambda\beta$ of the untyped lambda calculus with beta reduction.
This signature has only one inequation $\beta$ which expresses beta reduction at the root of a term,
   \[ \lambda x . M(N) \enspace \leadsto \enspace M [x:=N] \enspace .\]  
The initial representation of $\Lambda\beta$ is given by the untyped lambda calculus, equipped with
the reflexive and transitive beta reduction relation $\twoheadrightarrow_{\beta}$ as presented by \textcite{barendregt_barendsen}.

\subsection{Related Work}\label{sec:rel_work}

Initial Semantics results for syntax with variable binding were first presented 
on the LICS'99 conference.
Those results are concerned only with the \emph{syntactic aspect} of languages: they characterize the \emph{set of terms}
of a language as an initial object in some category, while not taking into account reductions on terms.
In lack of a better name, we refer to this kind of initiality results as \emph{purely syntactic}.

Some of these initiality theorems have been extended to also incorporate semantic aspects, e.g., in form of 
equivalence relations between terms. These extensions are reviewed in the second paragraph.

\paragraph{Purely syntactic results} 

Initial Semantics for ``pure'' syntax --- i.e.\ without considering semantic aspects --- 
with variable binding were presented by several people independently, differing in the modelling of 
variable binding:

The \emph{nominal approach} by \textcite{gabbay_pitts99} (see also \cite{gabbay_pitts, pitts}) 
uses a set theory enriched with \emph{atoms} to establish an initiality result. 
Their approach models lambda abstraction as a constructor which takes a pair of a variable name and a term as arguments.
In contrast to the other techniques mentioned in this list, in the nominal approach
syntactic equality is different from $\alpha$--equivalence.
\textcite{hofmann} proves an initiality result modelling variable binding in a 
Higher--Order Abstract Syntax (HOAS) style.
\textcite{fpt} (also \cite{fio02, fio05}) 
model variable binding through nested datatypes as introduced by \textcite{BirdMeertens98:Nested}.
\citeauthor{fpt}'s approach is extended to \emph{simply--typed syntax} by \textcite{DBLP:conf/ppdp/MiculanS03}.
\textcite{Tanaka:2005:UCF:1088454.1088457}
generalize and subsume those three approaches to a general category of contexts. 
An overview of this work and references to more technical papers is given by \textcite{Power:2007:ASS:1230146.1230276}.
\textcite{DBLP:conf/wollic/HirschowitzM07} prove an initiality result for untyped syntax based on 
the notion of \emph{module over a monad}. Their work has been extended to simply--typed syntax by \textcite{ju_phd}.

\paragraph{Incorporating Semantics}

Rewriting in nominal settings has been examined by \textcite{fernandez_gabbay_nominal_rewriting}.
\textcite{DBLP:journals/njc/GhaniL03} present rewriting for algebraic theories without variable binding;
they characterize 
equational theories (with a \emph{symmetry} rule) resp.\ rewrite systems (with \emph{reflexivity} and \emph{transitivity} rule, but without \emph{symmetry})
 as \emph{coequalizers} resp.\ \emph{coinserters} in a category of monads on
the categories $\Set$ resp.\ $\PO$. 
\textcite{DBLP:conf/icalp/FioreH07} have extended Fiore's work to integrate semantic aspects into initiality results.
In particular, \citeauthor{hur_phd}'s thesis \parencite{hur_phd} is dedicated to \emph{equational} systems for syntax with variable binding.
In a ``Further research'' section \parencite[Chap.\ 9.3]{hur_phd}, \citeauthor{hur_phd} suggests the use of preorders, or more generally, 
arbitrary relations to model \emph{in}equational systems.
\textcite{DBLP:conf/wollic/HirschowitzM07} prove initiality of the set of lambda terms modulo beta and eta conversion
in a category of \emph{exponential monads}.
In an unpublished paper \parencite{journals/corr/abs-0704-2900} they define a notion of \emph{half--equation} and 
\emph{equation} to express congruence between terms. We adopt their definition in this paper, but interpret a pair
of half--equations as \emph{in}equation rather than equation.
This emphasizes the \emph{dynamic} viewpoint of reductions as \emph{directed} equalities rather than
the \emph{static}, mathematical viewpoint one obtains by considering symmetric relations.
In a ``Future Work'' section, \textcite[Sect.\ 8]{DBLP:journals/iandc/HirschowitzM10} 
mention the idea of using preorders as an approach to model  semantics, 
and they suggest interpreting the untyped lambda calculus with beta and eta reduction rule as a monad over the category $\PO$ of preordered sets.
The present work gives an alternative viewpoint to their suggestion by considering
the lambda calculus with beta reduction --- and a class of programming languages in general --- 
as a preorder--valued \emph{relative} monad on the functor $\Delta:\Set\to\PO$.
The rationale underlying our use of relative monads from sets to preorders is that we consider \emph{contexts} 
to be given by unstructured sets, whereas \emph{terms} of a language carry
structure in form of a reduction relation. In this view it is reasonable to suppose variables and terms to live in \emph{different} 
categories, which is possible through the use of relative monads on the functor $\Delta:\Set\to\PO$ (cf.\ Def.\ \ref{def:delta}) 
instead of traditional monads (cf.\ also Rem.\ \ref{rem:endo_ord}).
Relative monads were introduced by \textcite{DBLP:conf/fossacs/AltenkirchCU10}. In that work, the authors
characterize the untyped lambda calculus as a relative monad over the inclusion functor from finite sets to sets.
Their point of view can be combined with ours, leading to considering monads on the functor
 $\comp{i}{\Delta}:\Fin\to\PO$, cf.\ Rem.\ \ref{rem:finite_contexts}.
T.\ \textcite{HIRSCHOWITZ:2010:HAL-00540205:2}, taking the viewpoint of Categorical Semantics, 
defines a category \textsf{Sig} 
of 2--signatures for \emph{simply--typed} syntax with reduction rules, and constructs an adjunction
between \textsf{Sig} and the category $\mathsf{2CCCat}$
of small cartesian closed 2--categories. He thus associates to any signature a 2--category of 
types, terms and reductions satisfying a universal property. 
More precisely, terms are given by morphisms in this category, and 
reductions are expressed by the existence of 2--cells between terms.
His approach differs from ours in the way in which variable binding is modelled:
Hirschowitz encodes binding in a Higher--Order Abstract Syntax (HOAS) style through exponentials.

\subsection{Synopsis}
In the first section we review the definition of (relative) monads 
and define modules over those monads as well as their morphisms.  
Some constructions on monads and modules are given, which will be of importance in what follows.

\noindent
In the second section we define arities, half--equations and inequations, as well as their representations.
Afterwards we state our main result. 
The running example in the first two sections is the 2--signature $\Lambda\beta$ of the lambda calculus with beta reduction.

\noindent
In the third section we describe some elements of the formalization of the main theorem in the proof assistant \textsf{Coq}.
 
\noindent
Some conclusions and future work are stated in the last section.

%% file: monad.tex
\section{Relative Monads \& Modules}\label{mon_mod}

This section presents the category--theoretic structures from which 
our models of a signature are built.

At first we review relative monads as defined by \textcite{DBLP:conf/fossacs/AltenkirchCU10}
and define modules for those monads (cf.\ Sec.\ \ref{sec:monads_rel_modules}).
Afterwards (cf.\ Secs.\ \ref{mod_examples} and \ref{subsec:deriv}) we port some constructions
from modules over (traditional) monads, 
as presented, for instance, by \textcite{DBLP:conf/wollic/HirschowitzM07}, \textcite{ju_phd}
and \textcite{ahrens_zsido}, 
to modules over relative monads.

\subsection{Modules over Relative Monads}\label{sec:monads_rel_modules}

We review the definition of relative monad as given by \textcite{DBLP:conf/fossacs/AltenkirchCU10}
and define suitable morphisms for them.
As an example, we consider the lambda calculus as a relative monad from sets to preorders, on the functor $\Delta :\Set\to\PO$.
Finally, we define \emph{modules} over relative monads and morphisms of such modules.

\begin{definition}[Relative Monad] 
Given categories $\C$ and $\D$ and a functor $F : \C\to \D$, a \emph{(relative) monad} $P : \C \stackrel{F}{\to}\D$ on $F$
is given by
\begin{itemize}
 \item a map $P\colon \C\to\D$ on the objects of $\C$,
 \item 
for each object $c$ of $\C$, a morphism $\eta_c\in \D(Fc,Pc)$ and
 \item 
for all objects $c$ and $d$ of $\C$ a \emph{Kleisli} operation 
 \[\sigma_{c,d}\colon \D (Fc,Pd) \to \D(Pc,Pd)\]
\end{itemize}
such that the following diagrams commute for all suitable morphisms $f$ and $g$:
\begin{equation*}
\begin{xy}
\xymatrix @=3pc{
Fc \ar [r] ^ {\we_c} \ar[rd]_{f} & Pc \ar[d]^{\kl{f}} \\
{} & Pd 
}
 \end{xy}\qquad
\begin{xy}
 \xymatrix@=3pc{ Pc \ar@/^1pc/[rd]^{\kl{\we_c}} \ar@/_1pc/[rd]_{\id}& {}\\ {} & Pc 
}
\end{xy}\qquad
 \begin{xy}
\xymatrix @=3pc{
 Pc \ar[r]^{\kl{f}} \ar[rd]_{\kl{\comp{f}{\kl{g}}}} & Pd\ar[d]^{\kl{g}} \\
  {} & Pe .\\
}
\end{xy}
\end{equation*}
 Here and later we omit the object indices of the Kleisli operation.
\end{definition}

\begin{remark}
A monad $P$ is equipped with a functorial structure (also denoted by $P$) by setting, for a morphism $f : a \to b$ in $\C$, 
\[
   P(f) := \lift_P(f) := \kl{\comp{Ff}{\we}} \enspace ,
\]
the functoriality axioms being a consequence of the monad axioms. 
\end{remark}

We are mainly interested in monads on a specific functor:

\begin{definition}[$\Delta:\Set\to\PO$]\label{def:delta}
 We call $\Delta: \Set\to\PO$ the functor from sets to preordered sets which associates  
 to each set $X$ the set itself together with the smallest preorder, i.e.\ the diagonal of $X$,
\[ \Delta(X) := (X,\delta_X) \enspace . \] 
The functor $\Delta$ is a full \emph{embedding}, i.e.\ it is fully faithful
and injective on objects.
Furthermore it is left adjoint to the forgetful functor $U:\PO\to\Set$,

\begin{equation*}   
       \begin{xy}
        \xymatrix@C=4pc{
                 **[l] \Set \rtwocell<5>_U^{\Delta}{'\bot} & **[r]\PO
}
       \end{xy} \enspace ,
\end{equation*}
 that is, the embedding $\Delta:\SET\to\PO$ is a coreflection.
We denote by $\varphi$ the family of isomorphisms
   \[   \varphi_{X,Y} : \PO(\Delta X, Y) \cong \Set(X,UY) \enspace . \]
We omit the indices of $\varphi$ whenever they can be deduced from the context.
\end{definition}

\begin{lemma}[Monads over $\Delta$ and Monads on $\Set$]\label{lem:rmon_delta_endomon}
  Let $P$ be a monad on $\Delta$.
 By postcomposing with the forgetful functor $U:\PO\to \Set$ we obtain a monad 
   \[\bar P : \Set\to \Set\enspace . \]
   The substitution is defined, for $f \in \SET(a,UPb)$ by setting
   \[ \sigma^{\bar P}(f):=  U \left(\kl{\varphi^{-1}f}\right) \enspace ,  \]
making use of the adjunction $f\mapsto \varphi^{-1}f \in \PO(\Delta a , Pb)$ of Def.\ \ref{def:delta}.
Conversely, to any monad $Q$ on $\Set$ we associate a relative monad by postcomposing with the functor $\Delta$:
  \[\Delta Q  : \Set\stackrel{\Delta}{\to}\PO  \enspace . \]
\end{lemma}

\noindent
The above construction in Lem.\ \ref{lem:rmon_delta_endomon} actually is an instance of a general lemma: 
given an adjunction $F\dashv G$ such that $GF = \Id_C$, one can associate 
a relative monad on $F$ to any monad on $\C$ and vice versa.
These maps are functorial, and one obtains an adjunction between a suitable category
of relative monads on $F$ and a category of monads on $\C$.
Details will be reported elsewhere.

\begin{example}\label{ex:ulcbeta}
Consider the set of all lambda terms indexed by their set of free variables
\begin{align*} 
\ULC (V) ::=\quad &\Var : V \to \ULC(V) \\ 
        {}\mid{} &\Abs : \ULC (V') \to \ULC (V) \\ 
        {}\mid{} &\App : \ULC(V) \times \ULC(V)\to \ULC(V) \enspace ,
\end{align*}
where $V' := V + \lbrace * \rbrace$ is the set $V$ enriched with a new distinguished element, 
i.e.\ a context extended by one additional free variable.
We occasionally write $\lambda$ for $\Abs$ and denote application by juxtaposition.
\textcite{alt_reus}
interpret $\ULC$ as a monad over sets.
\textcite{DBLP:conf/fossacs/AltenkirchCU10} interpret $\ULC$ as a relative monad on the inclusion functor
$i:\Fin\to\Set$, by restricting contexts to be given by \emph{finite} sets (cf.\ also Rem.\ \ref{rem:finite_contexts}).

We equip each $\ULC(V)$ with a preorder taken as the reflexive--transitive closure of the 
relation generated by the rule
\[ \tag{$\beta$}  \lambda M (N) \enspace \leq \enspace M [*:= N] \enspace  \]
and its propagation into subterms.
 This defines a monad 
  from sets to preorders
    \[\ULCB : \SET\stackrel{\Delta}{\to}\PO.\] 
 The family $\we^{\ULC}$ is given by the constructor $\Var$, and the Kleisli operation 
  \[\sigma_{X,Y} : \PO\bigl(\Delta X, \ULCB(Y)\bigr) \to \PO\bigl(\ULCB(X),\ULCB(Y)\bigr)\]
 is given by simultaneous substitution. 
Via the adjunction of Def.\ \ref{def:delta} the substitution can also be read as
\begin{equation*}\sigma_{X,Y} : \SET\bigl(X, \ULC(Y)\bigr) \to \PO\bigl(\ULCB(X),\ULCB(Y)\bigr) \enspace .
\end{equation*}
The substitution can hence be chosen as for the monad $\ULC$, but one has to prove the additional property of monotonicity in
the first--order argument.
\end{example}

Inspired by the example above we sometimes call the Kleisli operation of a monad \emph{monadic substitution}, 
even for monads that do not denote the terms of a language over free variables.

For two monads $P$ and $Q$ from $\C$ to $\D$ a \emph{morphism of monads} is a
family of morphisms $\tau_c\in\D(Pc,Qc)$
that is compatible with the monadic structure:

\begin{definition}[Morphism of Relative Monads]\label{def:morph_of_relmons}
Given two relative monads $P$ and $Q$ from $\C$ to $\D$ on the same functor $F\colon\C\to\D$, 
 a \emph{morphism of monads} from $P$ to $Q$ is given by a collection of morphisms $\tau_c\in \D(Pc,Qc)$ 
such that the following diagram commutes for any suitable morphism $f$:
\begin{equation*}
 \begin{xy}
  \xymatrix @=3pc{
  Pc \ar[r]^{\kl[P]{f}} \ar[d]_{\tau_c}& Pd \ar[d]^{\tau_d} & Fc \ar[r]^{\we^P_c} \ar[rd]_{\we^Q_c} & Pc \ar[d]^{\tau_c} \\
  Qc \ar[r]_{\kl[Q]{\comp{f}{\tau_d}}} & Qd & {} & Qc\\
}
 \end{xy}
\end{equation*}
As a consequence from these commutativity properties the family $\tau$ is a natural transformation between the functors induced by the monads $P$ and $Q$.
\end{definition}

Monads on $F:\C\to\D$ and their morphisms form a category $\RMon{F}$  where 
identity and composition of morphisms are defined by pointwise identity and composition of morphisms. 
We have a similar category of traditional monads:
\begin{definition}[Category of Monads]
 We denote by $\Mon{\C}$ the category of monads on $\C$ and morphisms of such monads.
 More precisely, a morphism $f: P\to Q$ in $\Mon{\C}$ is given by a family 
 $\tau_c : \C(Pc,Qc)$
 of morphisms that is compatible with the monadic structure, analogously to the diagrams of Def.\ \ref{def:morph_of_relmons}.
\end{definition}

\begin{lemma}[Adjunction between $\RMon{\Delta}$ and $\Mon{\Set}$]\label{lem:adjunction_rmon_delta_mon}
 The maps defined in Lem.\ \ref{lem:rmon_delta_endomon} 
     give rise to 
  an adjunction between the category of monads over $\Delta$ and the category of monads over sets,
  where the functor $U_*$ is defined on objects as in Lem.\ \ref{lem:rmon_delta_endomon} by $U_*(R):=\bar{R}$:
\begin{equation*}   
       \begin{xy}
        \xymatrix@C=4pc{
                   **[l]\Mon{\Set} \rtwocell<5>_{U_*}^{\Delta_*}{'\bot} & **[r]\RMon{\Delta}
}
       \end{xy} \enspace .
\end{equation*}
\end{lemma}

\noindent
Informally, the notion of relative monad on a base functor $F$ 
is obtained from the notion of monad (in Manes style, i.e.\ with Kleisli operation) 
by inserting applications of the functor $F$ where necessary \parencite{DBLP:conf/fossacs/AltenkirchCU10}.
Similarly, one obtains the notion of module over a relative monad from the notion of module over a monad
--- in form of a Kleisli operation as presented, for instance, by \textcite{ahrens_zsido} ---  by inserting applications of the base functor $F$:

\begin{definition}[Module over Relative Monad]
Let $P\colon\C\stackrel{F}{\to}\D$ be a relative monad and let $\E$ be a category. A \emph{relative module $M$ over $P$ with codomain $\E$} is given by
\begin{itemize}
 \item a map $M\colon \C \to \E$ on the objects of the categories involved and 
 \item for all objects $c,d$ of $\C$ a map 
      \[ 
          \varsigma_{c,d}\colon \D (Fc,Pd) \to \E (Mc,Md)
      \]
\end{itemize}
such that the following diagrams commute for all suitable morphisms $f$ and $g$:
\begin{equation*}
\begin{xy}
\xymatrix @=3pc{
 Mc \ar[r]^{\mkl{f}} \ar[rd]_{\mkl{\comp{f}{\kl{g}}}} & Md\ar[d]^{\mkl{g}} & Mc  \ar@/^1pc/[rd]^{\mkl{\we_c}} \ar@/_1pc/[rd]_{\id} & {} \\
  {} & Me & {} & Mc . 
}
\end{xy}
\end{equation*}
\end{definition}

\noindent
Functoriality for such a module $M$ is defined in a similar way as for monads, by setting,
  for any morphism $f : c \to d$ in $\C$,
 \[ M(f) := \rmlift_M(f) := \varsigma({ \comp{Ff}{\eta}}) \enspace .\]

\begin{example}[Monoids and Monoid Actions]
 An anonymous referee suggested the following example: let $F : 1 \to \Set$ be the functor on the one--object category
  which maps to the final object, i.e.\ the singleton set.
 Then a relative monad on $F$ is a monoid.
 Given a relative monad $G$ on $F$, a $G$--module with codomain $\Set$ is a monoid action of $G$.
\end{example}

The following examples of modules are instances of more general constructions explained in the next section:

 \begin{example}[Tautological Module]\label{ex:ulcb_taut_mod}
  The map $\ULCB : V \mapsto \ULCB(V)$ yields a module over the monad $\ULCB$, the \emph{tautological module} $\ULCB$.
 \end{example}
 
 \begin{example}[Derived Module]\label{ex:ulcb_der_mod}
  Let $V' := V + \{*\}$. The map $\ULCB' : V \mapsto \ULCB(V')$ inherits the structure of an $\ULCB$--module from the
   tautological $\ULCB$--module $\ULCB$ of Ex.\ \ref{ex:ulcb_taut_mod}.
   We call $\ULCB'$ the \emph{derived module} of the module $\ULCB$. 
 \end{example}
 \begin{example}[Product of Modules]\label{ex:ulcb_prod_mod}
  The map $V\mapsto \ULCB(V)\times\ULCB(V)$ inherits an $\ULCB$--module structure from $\ULCB$.
 \end{example}

A \emph{module morphism} is given by a family of morphisms that is compatible with module substitution:
\begin{definition}[Module Morphism]
 Let $P$ be a relative monad on $F:\C\to\D$, and 
 let $M$ and $N$ be two relative modules over $P$ with codomain $\E$. 
A \emph{morphism of relative $P$--modules} from $M$ to $N$ is given by a collection of morphisms $\rho_c\in\E(Mc,Nc)$ 
such that for any morphism $f\in \D(Fc,Pd)$ the following diagram commutes:
\begin{equation*}
 \begin{xy}
  \xymatrix @=3pc{
  Mc \ar[r]^{\mkl[M]{f}} \ar[d]_{\rho_c}& Md \ar[d]^{\rho_d}  \\
  Nc \ar[r]_{\mkl[N]{f}} & Nd \\
}
 \end{xy}
\end{equation*}
\end{definition}

\noindent
The modules over $P$ with codomain $\E$ and morphisms between them form a category called $\RMod (P,\E)$. 
Note that the ``R'' here stands for ``relative'', not for ``right'' as opposed to ``left''.
Composition and identity morphisms of modules are defined by pointwise composition and identity, similarly to the category of monads.

\begin{example}[Exs.\ \ref{ex:ulcb_taut_mod}, \ref{ex:ulcb_der_mod}, \ref{ex:ulcb_prod_mod} cont.]\label{ex:ulcb_constructor_mod_mor}
 Abstraction and application are morphisms of $\ULCB$--modules,
  \begin{align*} \Abs &: \ULCB' \to \ULCB \enspace , \\
                \App &: \ULCB \times \ULCB \to \ULCB \enspace .
  \end{align*}
\end{example}

\subsection{Constructions on Monads and Modules}\label{mod_examples}

The constructions on modules over monads as used by 
 \textcite{DBLP:conf/wollic/HirschowitzM07} 
 carry over to relative monads:

\begin{definition}[Tautological Module]
 Given a monad $P$ on $F:\C\to\D$, 
 we define the \emph{tautological module} (also denoted by $P$) over $P$ 
 to be the module $(M, \varsigma) := (P,\sigma)$, i.e.\ with object map $P$ and module substitution
  given by the monad substitution.
  Thus the monad $P$ can be considered as an object in the category $\RMod(P,\D)$.
\end{definition}

\begin{definition}[Constant and terminal module] \label{def:const_mod}
Let $P$ be a monad on $F:\C\to\D$.
For any object $e \in \E$ the constant map $T_e\colon\C\to\E$, $c\mapsto e$ for all $c\in \C$ 
yields a $P$--module $(T_e, \id)$.
In particular, if $\E$ has a terminal object $1_\E$, then the constant module $(T_{1_{\E}},\id)$ is terminal in $\RMod(P,\E)$.
\end{definition}

\begin{definition}[Postcomposition with a functor]
 Let $P$ be a monad on $F:\C\to\D$, and let $M$ be a $P$--module with codomain $\E$.
 Let $G: \E \to \X$ be a functor. Then the object map $\comp{M}{G}:\C\to \X$ defined by $c\mapsto G(M(c))$ 
 is equipped with a $P$--module structure by setting, for $c, d \in \C$ and $f\in \D(Fc,Pd)$,
  \[  \varsigma^{\comp{M}{G}}(f) := G(\varsigma^M(f)) \enspace . \]
 For $M:=P$ and $G$ a constant functor mapping to an object $x\in \X$ and its identity morphism $\id_x$, 
 we obtain the constant module $(T_x,\id)$ as in Def.\ \ref{def:const_mod}. 
\end{definition}

Let $P$ and $Q$ be two monads on $F : \C \to \D$. 
Given a monad morphism $h : P \to Q$, we can turn any 
$Q$--module $M$ into a $P$--module,
by ``pulling it back'' along $h$:

\begin{definition}[Pullback module] 
Let $h \colon P \to Q$ be a morphism of monads on $F:\C\to\D$ 
and let $M$ be a $Q$--module with codomain $\E$. We define a $P$--module $h^* M$ to $\E$ with object map $c\mapsto M c$ 
by setting 
\[\mkl[h^*M] f := \mkl[M]{\comp {f}{h_d}} \enspace . \] 
This module is called the \emph{pullback module of $M$ along $h$}. The pullback extends to module morphisms and is functorial.
\end{definition}

\begin{remark}
 The pullback $P$--module $h^*M$ has the same underlying functor as the $Q$--module $M$. It is merely the substitution action that 
 changes: while $\varsigma^M$ expects morphisms in $\C(Fc,Qd)$ as arguments, the substitution of $h^*M$ expects morphisms
 in $\C(Fc,Pd)$.
\end{remark}

\begin{definition}[Induced module morphism] \label{def:ind_mod_hom}
With the same notation as before, the monad morphism $h$ induces a morphism of $P$--modules $h\colon P \to h^*Q$.
\end{definition}

\begin{remark}
 Note that the preceding two constructions do not change the functor resp.\ natural transformation underlying the
 module resp.\ morphism of modules. 
\end{remark}

The following construction explains Ex.\ \ref{ex:ulcb_prod_mod}:
\begin{definition}[Products] 
Suppose the category $\E$ is equipped with a product. Let $M$ and $N$ be $P$--modules with codomain $\E$. Then the map (on objects)
\[ \C\to\E, \quad c \mapsto Mc \times Nc \] 
 is equipped with a module substitution by setting
 \[  \mkl[M\times N]{f} := \mkl[M]{f} \times \mkl[N]{f} \enspace . \]
This construction extends to a product on $\RMod(P,\E)$.
\end{definition}

There is also a category where modules over different monads are grouped together:
\begin{definition}\label{def:lmod}
Let $\C,\D$ and $\E$ be categories and $F:\C\to\D$ be a functor. 
We define the category $\LMod(F,\E)$ (``L'' for ``large'') to be the category whose objects are pairs $(R,M)$, 
where $R$ is a monad over $F : \C \to \D$ and $M$ is an $R$--module.
A morphism from $(R,M)$ to $(S,N)$ is a pair $(\rho, \tau)$ where $\rho : R \to S$ is a monad morphism and $\tau$ is an
$R$--module morphism $\tau : M \to \rho^* N$. 
\end{definition}

\vspace{1.5ex}

We are particularly interested in monads over the functor $\Delta: \SET\to\PO$.
The following construction --- \emph{derivation} --- applies to modules over such monads. 

\subsection{Derived Modules}\label{subsec:deriv}

Roughly speaking, a binding constructor makes free variables disappear. 
Its input are hence terms in an extended context, i.e.\ with (one or more) additional free variables 
compared to the output. 
\emph{Derivation} is about context extension.

Formally, given a set $V$ ($V$ for variables), we consider a new set \[V' := V + \{*\}\]
which denotes $V$ enriched with a new distinguished element -- the ``fresh'' variable. 
The map $V\mapsto V'$ can be extended to a monad on the category of sets and is hence functorial.
Given a map $f:V\to W$ and $w\in W$, we call 
    \[ \defaultmap(f,w) : V' \to W\] 
the coproduct map defined by
 
 \[ \defaultmap(f,w) = [f,\lambda x.w]  \enspace . \]

\begin{definition}[Derived Module]\label{def:deriv_one}
Given a monad $P$ on $\Delta:\SET\to\PO$ and a $P$--module $M$ with codomain $\E$,
we define the \emph{derived module} by setting
\[  M'(V) := M (V') \enspace . \]
For a morphism $f\in\PO(\Delta V,PW)$ the module substitution for the derived module is given by
\[ \mkl[{M'}]{f} := \mkl[M]{\mathrm{shift}(f)} \enspace . \]
Here the shifted map \[\mathrm{shift}(f) \in \PO(\Delta(V'),P(W'))\] 
is defined via the adjunction of Def.\ \ref{def:delta} as 
\[ \mathrm{shift}(f) := \varphi^{-1}\Bigl(\defaultmap\bigl(\comp{f}{P(\inl)},\we(\inr(*))\bigr)\Bigr) \enspace ,  \]
where $[\inl,\inr] = \id_{W'}$.
\end{definition}
\noindent
Derivation extends to an endofunctor on the category of $P$--modules with codomain $\E$. 

\begin{remark}
When $P$ is a monad of terms over free variables, 
the map $\mathrm{shift}f$ sends the additional variable of $V'$ to $\eta^{P}(*)$, 
i.e.\ to the term consisting of just the ``freshest'' free variable.
When recursively substituting with a map $f : V \to P(W)$, terms under a binder such as $\lambda$ must be substituted 
with the map $\mathrm{shift}(f)$.
\end{remark}

\begin{definition}\label{def:alg_notation}
Given a natural number $n$, we write $M^n$ for the module $M$ derived $n$ times.
Given a list $s = [n_1,\ldots,n_m]$ of natural numbers, we write $M^s := M^{n_1} \times \ldots \times M^{n_m}$.

Product and derivation are functorial, and we use the same notation \emph{on morphisms}. That is,
given a morphism of $P$--modules $\rho : M \to N$,
we write \[\rho^s := \rho^{n_1} \times \ldots \times \rho^{n_m} : M^s \to N^s \enspace . \]
 
\end{definition}

\noindent
The pullback operation commutes with products and derivations :

\begin{lemma} \label{pb_prod}Let $\C$ and $\D$ be categories and $\E$ be a category with products. Let $P$ and $Q$ 
be monads
on $F:\C\to\D$ and $\rho : P \to Q$ a monad morphism. 
Let $M$ and $N$ be $P$--modules with codomain $\E$. The pullback functor is cartesian:
 \[ \rho^* (M \times N) \cong \rho^*M \times \rho^*N \enspace .\]
\end{lemma}
\begin{lemma} \label{pb_comm} 
 Let $P$ be a relative monad on $\Delta:\Set\to\PO$ and $M$ a module over $P$ with codomain $\E$. 
 Then we have
\[ \rho^* (M') \cong (\rho^*M)' \enspace . \]
\end{lemma}

\begin{definition}[Substitution of \emph{one} variable]\label{ex:subst}
 We denote by $\PS$ (``w'' for ``weak'') the category whose objects are preordered sets, and
 where a morphism from $A$ to $B$ is a map between the underlying sets, 
 that is, a morphism $UA \to UB$ in $\Set$.
 
Given a monad $P : \SET\stackrel{}{\to}\PO$ on $\Delta$, and a $P$--module $M$ with codomain $\PO$, 
we can consider $M$ as a $P$--module with codomain $\PS$ by postcomposing with the injection.
 We denote this module by $\hat{M} \in \RMod(P,\PS)$. 
 For any set $X$, we define the substitution of just one variable,
 \begin{align*} \subst_X : P(X') \times P(X) &\to P(X) \enspace, \\  
                              (y,z) \enspace &\mapsto \enspace y [*:= z] := \kl{\defaultmap(\we_X , z)}(y) \enspace . 
 \end{align*}
 This defines a morphism of $P$--modules with codomain $\PS$,
\[\subst^P :  \hat{P}' \times \hat{P}\to \hat{P} \enspace .\] 
\end{definition}

\begin{remark}
 Note that the substitution module morphism defined above is by construction monotone in
 its first argument, but not in its second argument.
 This is the reason why we cannot consider $\subst^P$ as a morphism of $P$--modules
 \[ \subst^P : P' \times P \to P \enspace , \]
but have to switch to the category $\PS$.
 This fact and a way to ensure monotonicity also in the second argument 
      are explained more generally in Rem.\ \ref{rem:about_substitution}.
\end{remark}

%% file: signature.tex
\section{2--Signatures and their Representations}\label{sig_rep}

In this section we define \emph{2--signatures} and their representations.
We then prove that, given any 2--signature, its associated category of representations has an initial object, the 
language generated by the 2--signature.

An \emph{arity} describes the number of arguments and binding behaviour of a constructor of a syntax.
A \emph{1--signature} $S$ is a family of arities and as such specifies the terms of a language.
An \emph{inequation over $S$} --- also called \emph{$S$--inequation} --- expresses a reduction rule on 
the terms of the syntax associated to $S$.  
A \emph{2--signature} is given by  a 1--signature $S$ and a set of $S$--inequations. 

We define \emph{representations} of a 1--signature $S$ analogously to  \textcite{DBLP:conf/wollic/HirschowitzM07} and \textcite{ju_phd},
except that we use \emph{relative} monads and modules over such monads.
Afterwards we use  \emph{inequations over $S$} to specify reduction rules, and we consider those representations of $S$
that satisfy the given inequations.
We show that among those representations there is an initial representation, which integrates the terms 
generated by $S$, equipped with reductions according to the given inequations.
Our inequations are precisely \citeauthor{journals/corr/abs-0704-2900}'s (\citeyear{journals/corr/abs-0704-2900}) \emph{equations},
i.e.\ parallel pairs of half--equations.
We simply interpret such a pair to define a \emph{(directed) reduction} rather than an equality.
Throughout this section the running example is the 2--signature $\Lambda\beta$ which specifies 
the lambda calculus with beta reduction.

\subsection{Arities, 1--Signatures and their Representations}

We first consider pure syntax, i.e.\ syntax without reductions.
We specify syntax by a \emph{1--signature}:

\begin{definition}[Arity, 1--Signature]\label{def:1--signature}
 An \emph{arity} is a list of natural numbers.
A \emph{1--signature} is a family of arities.
\end{definition}

Intuitively, the length of an arity specifies the number of arguments of a constructor, and
the $i$--th entry of the arity specifies the number of variables which are bound by
the constructor in the $i$--th argument. 

\begin{example} \label{ex:sig_ulc}
The 1--signature $\Lambda$ of untyped lambda calculus is given by the two arities
\[ \app := [0,0] \enspace , \quad \abs := [1] \enspace . \]
\end{example}

\begin{definition}
 Let $s := [n_1,n_2,\ldots,n_m]$ be an arity and $P$ be a monad on the functor $\Delta:\Set\to\PO$.
 We call
     \[ \dom(s,P) := P^s = P^{n_1} \times \ldots \times P^{n_m} \]
 the \emph{domain module} of $s$ for $P$.
 Note that we use the notation defined in Def.\ \ref{def:alg_notation}.
\end{definition}

\begin{definition}[Representations of a 1--Signature]\label{def:rep_in_rmonad}
A \emph{representation $R$ of a 1--signature} $S$ is given by a monad $P$ over the functor $\Delta:\Set\to\PO$ 
and, for each arity $s = [n_1,n_2,\ldots,n_m] \in S$, a morphism of $P$--modules 
\[s^R : \dom(s,P) \to P \enspace . \]
 Given a representation $R$, we denote its underlying monad by $R$ as well. 
\end{definition}

\begin{example} [Ex. \ref{ex:sig_ulc} continued] \label{ex:ulc_rep}
A representation $R$ of the 1--signature $\Lambda$ is given by
\begin{itemize}
 \item a monad $R : \SET\stackrel{\Delta}{\to}\PO$ and
 \item two morphisms of $R$--modules in $\RMod(R,\PO)$,
       \[ \app^R : R \times R \to R \quad\text{and}\quad \abs^R : R' \to R \enspace .\] 
\end{itemize}

\end{example}

\begin{remark}
 A representation of a 1--signature \`a la \textcite{DBLP:conf/wollic/HirschowitzM07}
 is defined analogously, except for the use of (plain) monads on the category of sets and modules thereon
 instead of \emph{relative} monads and modules on relative monads.
\end{remark}

\begin{definition}\label{lem:rep_endo_rep_rel}
 To any representation $R$ of a 1--signature $S$ in a relative monad $R$ as defined in Def.~\ref{def:rep_in_rmonad} 
we associate a representation $U_*(R)$ of $S$ in the monad $\bar{R}$ (cf.\ Lem.\ \ref{lem:rmon_delta_endomon}) 
in the sense of \textcite{DBLP:conf/wollic/HirschowitzM07} 
by postcomposing the representation module morphism of any arity $s$ of $S$ with the forgetful functor from preorders to sets:

\[ s^R : R^s \to R \quad \mapsto \quad s^{U_*R} :  {\bar R}^s \to \bar R \enspace , \]
where we use $\bar {(R ^s)} = ({\bar R})^s$.
  Conversely, to any representation $Q$ of $S$ in a monad $Q$ over sets we can associate a representation $\Delta_{*} Q$
 of $S$ in the relative monad $\Delta_* Q$ on $\Delta$, by postcomposing the representation module morphisms with $\Delta$.

\end{definition}

Morphisms of representations are monad morphisms which commute with the representation morphisms of modules:
\begin{definition}[Morphism of Representations]\label{def:morph_of_reps}
Let $P$ and $Q$ be representations of a signature $S$. A \emph{morphism of representations} $f: P \to Q$ 
is a morphism of monads $f: P\to Q$ such that the following diagram commutes for any arity $s$ of S: 
\[
\begin{xy}
\xymatrix @=3pc{
  P^s\ar[r]^{s^P} \ar[d]_{f^s} & P \ar[d]^{f}  \\
  f^* Q^s \ar[r]_{f^* s^Q} & f^*Q. \\
}
\end{xy}
\]
\end{definition}

\noindent
Note that we make extensive use of the notation defined in Def.\ \ref{def:alg_notation}.
To make sense of this diagram it is necessary to recall the constructions on modules of section \ref{mod_examples}. 
The diagram lives in the category $\RMod(P,{\PO})$. 
The vertices are obtained from the tautological modules $P$ resp.\ $Q$ over the monads $P$ resp.\ $Q$ by applying 
the derivation and pullback functors as well as by the use of the product in the category of $P$--modules into $\PO$.
The vertical morphisms are module morphisms induced by $f$, to which functoriality of derivation and products are applied. 
Furthermore, instances of Lems.\ \ref{pb_prod} and \ref{pb_comm} are hidden in the lower left corner. 
The lower horizontal morphism makes use of the functoriality of the pullback operation.

\begin{example}[Ex. \ref{ex:ulc_rep} continued]
  Let $P$ and $R$ be two representations of $\Lambda$. 
A \emph{morphism} from $P$ to $R$ is given by a morphism of monads
 $f : P \to R$ such that the following diagrams of $P$--module morphisms commute:
\begin{equation*}
\begin{xy}
 \xymatrix @C=4pc  {
   **[l]P\times P \ar[r]^{\app^P} \ar[d]_{f \times f}& P \ar[d]^{f} & & **[l] P' \ar[r]^{\abs^P} \ar[d]_{f'}& P\ar[d]^{f} \\
   **[l]f^*(R\times R) \ar[r]_{f^*(\app^R)} & f^*R & & **[l]f^* R' \ar[r]_{f^*(\abs^R)}& f^*R.
}
\end{xy}
\end{equation*}
\end{example}

\noindent
Composition and identity morphisms of representations are given by composition and identity
of monad morphisms. We obtain a category of representations:

\begin{definition}[Category of Representations]\label{lem:cat_of_1-reps}
 Representations of $S$ and their morphisms form a category $\Rep^{\Delta}(S)$.
 \end{definition}

Since we are not considering any reductions on terms yet, but only plain syntax, it comes as no surprise that,
for any 1--signature $S$, 
our category of representations of $S$ relates to \citeauthor{DBLP:conf/wollic/HirschowitzM07}'s:

\begin{lemma}\label{lem:adjunction_reps}
 The assignment of Def.\ \ref{lem:rep_endo_rep_rel} extends to an adjunction between our category of representations $\Rep^{\Delta}(S)$ in relative monads on $\Delta$ 
 and \citeauthor{DBLP:conf/wollic/HirschowitzM07}'s category $\Rep(S)$ of representations in monads over sets:

  \begin{equation*}   
       \begin{xy}
        \xymatrix@C=3pc{
                   **[l]\Rep(S)\rtwocell<5>_{U_*}^{\Delta_*}{'\bot} & **[r]\Rep^{\Delta}(S)
}
       \end{xy} \enspace . 
\end{equation*}
\end{lemma}

\begin{lemma}\label{lem:initial_wo_ineq}
  Given a signature $S$, the category of representations of $S$ in relative monads on $\Delta$ has an initial representation.
Its underlying monad associates to any set $V$ of variables the
set of terms of the language specified by $S$ in the context $V$, equipped with the diagonal preorder.
\end{lemma}
\begin{proof}
 This is a direct consequence of the fact that left adjoints preserve colimits \parencite[Chap.\ V.5]{maclane}, thus, in particular, initial objects.
\end{proof}

\begin{remark}
 The formalization in \textsf{Coq} of Lem.\ \ref{lem:initial_wo_ineq} (cf.\ Sec.\ \ref{subsec:1-sig-formal}) does not appeal to 
 \citeauthor{DBLP:conf/wollic/HirschowitzM07}'s result, but constructs the initial object from scratch.
\end{remark}

\subsection{Inequations over 1--Signatures}

Consider the beta rule of lambda calculus,

\[ \lambda M(N) \leadsto M[*:=N] \enspace . \]

\noindent
We would like to express such a rule through a suitable inequation.
In our formalism, abstraction and application are considered as morphisms of modules 
(cf.\ Ex.\ \ref{ex:ulcb_constructor_mod_mor}), 
and so is substitution (cf.\ Def.\ \ref{ex:subst}).
This suggest to define (in)equations over a 1--signature $S$ 
as \emph{parallel pairs} of module morphisms, indexed by representations of $S$.
Put differently, an (in)equation associates a parallel pair of module morphisms to any representation of $S$.
\textcite{journals/corr/abs-0704-2900} specify equations through such pairs of (indexed) module morphisms over (plain) monads.
We adapt their definition to our use of \emph{relative} monads and modules over such monads.
Afterwards we simply interpret a pair of half--equations as \emph{inequation} rather than equation.

\begin{definition}[Category of Half--Equations]\label{def:cat_halfeq}
Let $S$ be a signature. An \emph{$S$--module} $U$
is a functor from the category of representations of $S$ to the category of modules $\LMod (\Delta,\PS)$
commuting with the forgetful functor to the category of relative monads on $\Delta$:
\[
 \begin{xy}
  \xymatrix{
                \Rep^{\Delta}(S) \ar[rr]^{U} \ar[rd]_{} & {} & \LMod(\Delta,\PS) \ar[ld]\\
                     {}          & \RMon{\Delta} . & {}
}
 \end{xy}
\]

\noindent
We define a morphism of $S$--modules to be a natural transformation which
becomes the identity when composed with the forgetful functor. 
We call these morphisms \emph{half--equations}.
The collection of $S$--modules and their morphisms
yield a category 
which we call the category of $S$--modules
(or the category of half--equations).
\end{definition}

\begin{remark}
Objects of $\LMod(\Delta,\PS)$ are pairs of a monad $P$ on $\Delta:\Set\to\PO$ and a $P$--module $M$.
Given an $S$--module $U$, we sometimes write $U(R)$ for the second component of $U(R)$, i.e.\ for 
the module over the monad (underlying the representation) $R$, see for instance Rem.\ \ref{rem:half_equation_comm}.
We also write \[U^R_X := U(R)(X)\] for the value of an $S$--module at the representation $R$ and the set $X$.
Similarly, for a half--equation $\alpha : U \to V$ we write \[\alpha^R_X := \alpha(R)(X) : U^R_X \to V^R_X \enspace.\]

\end{remark}

\begin{remark}\label{rem:half_equation_comm}
 A half--equation $\alpha$ from $S$--module $U$ to $V$ associates to 
 any representation $P$ a morphism of $P$--modules $\alpha^P : U(P) \to V(P)$ in $\RMod(P,\PS)$ such that
for any morphism of $S$--representations $f : P \to R$ the following diagram commutes:
 \begin{equation*}
\begin{xy}
 \xymatrix @=4pc  {
   **[l](P,U(P)) \ar[r]^{\alpha^P} \ar[d]_{(f,U(f) )}& **[r](P,V(P)) \ar[d]^{(f, V(f))} \\
   **[l](R, f^*(U(R))) \ar[r]_{\alpha^R} & **[r](R,f^* (V(R)))
}
\end{xy}
\end{equation*}
\end{remark}

\begin{lemma}
 The category of $S$--modules is cartesian.
\end{lemma}

We give some examples of generic $S$--modules. The inductive class of $S$--modules
thus defined is of importance later.

\begin{definition}[Classic $S$--module]\label{def:alg_s_mod}
 We call \emph{classic} any $S$--module 
satisfying the following inductive predicate: 
 \begin{itemize}
  \item The map $\Theta : R \mapsto (R,\hat{R})$ is a classic $S$--module.
  \item If the $S$--module $M : R\mapsto (M_1(R),M_2(R))$ is classic, so is 
              \[M' : R\mapsto \bigl(M_1(R), M_2(R)'\bigr).\]
  \item If $M$ and $N$ are classic, so is \[M\times N : R\mapsto \bigl(M_1(R), M_2(R)\times N_2(R)\bigr).\]
  \item The terminal module $*: R\mapsto (R,1)$ is classic.
 \end{itemize}
Using the same notation as in Def.\ \ref{def:alg_notation}, any list of natural numbers defines uniquely a 
 classic $S$--module.
\end{definition}

The following examples of half--equations are building blocks for the inequation specifying 
beta reduction:

\begin{definition}\label{def:subst_half_eq}
 The substitution operation
 \[\subst  : R\mapsto \subst^R : \hat{R}' \times \hat{R}\to \hat{R} \]
 is a half--equation over any signature $S$.
 Its domain and codomain are classic.
\end{definition} 
\begin{example}[Ex. \ref{ex:sig_ulc} continued]\label{ex:app_circ_half}
  The map
 \[ \app \circ (\abs \times\id) : R \mapsto \app^R \circ (\abs^R \times \id^R) :  \hat{R}' \times \hat{R} \to \hat{R}\]
 is a half--equation over the signature $\Lambda$.
\end{example}

\begin{definition}
  \label{def:arity_classic_module_untyped}
 Any arity $s = [n_1, \ldots, n_m] \in S$ defines a classic $S$--module 
    \[\dom(s) : R\mapsto R^s \enspace . \] 
\end{definition}

An \emph{inequation} is a pair of parallel half--equations:

\begin{definition}[Inequation, 2--Signature]
 Given a 1--signature $S$, an \emph{$S$--inequation} is a pair of parallel half--equations between $S$--modules. 
We write 
     \[\alpha \leq \gamma : U\to V\] 
for the inequation $(\alpha, \gamma)$ with domain $U$ and codomain $V$.
A \emph{2--signature} is a pair $(S,A)$ of a 1--signature $S$ and a set $A$ of 
 $S$--inequations.
\end{definition}

Given a set $A$ of $S$--inequations, we can ask whether a given representation $R$ of $S$ satisfies the inequations of $A$.
If this is the case, we call $R$ a representation of the 2--signature $(S,A)$:

\begin{definition}[Representations of a 2--Signature] \label{def:rep_ineq}
 Let $\alpha\leq \gamma : U\to V$ be an inequation over $S$, and let $R$ be a representation of $S$.
 We say that \emph{$R$ satisfies $\alpha\leq \gamma$} if 
  $\alpha^R \leq \gamma^R$ pointwise, i.e.\ if for any set $X$ and any $y\in U(R)(X)$, 
    \[\alpha^R_X(y) \enspace  \leq \enspace \gamma^R_X(y) \enspace .\]

\noindent
For a set $A$ of $S$--inequations, we call \emph{representation of $(S,A)$} any representation of $S$ that
satisfies each inequation of $A$.
We define the
category of representations of the 2--signature $(S, A)$ to be the full subcategory in the category of
representations of $S$ whose objects are representations of $(S, A)$.

\end{definition}

\begin{example}[Ex. \ref{ex:app_circ_half} continued] \label{ex:sig_ulc_prop}
We denote by $\beta$ the $\Lambda$--inequation 
\[  \quad \app \circ (\abs \times\id) \leq \subst \enspace .  \tag{$\beta$}\]
We call $\Lambda\beta$ the 2--signature $((\app,\abs),\beta)$.
 A representation $P$ of $\Lambda\beta$ is given by
\begin{itemize}
 \item a monad $P : \SET\stackrel{\Delta}{\to}\PO$ and
 \item two morphisms of $P$--modules 
       \[ \app : P \times P \to P \quad\text{and}\quad \abs :  P' \to P\] 
\end{itemize}
 such that for any set $X$ and any $y\in P(X')$ and $z\in P X$
\begin{equation*} \label{eq:prop_arity}
\app_X (\abs_X (y), z) \enspace \leq \enspace y [*:=z] \enspace .  
\end{equation*} 
\end{example}

\subsection{Initiality for 2--Signatures}

Given a 2--signature $(S,A)$, we would like to exhibit an initial object in its associated
 category of representations.
However, we have to rule out inequations which are never satisfied, 
since an empty category obviously does not not have an initial object.
We restrict ourselves to inequations with a classic codomain:

\begin{definition}[Classic Inequation]\label{def:ineq_soft}
  We say that an $S$--inequation is \emph{classic} if its co\-do\-main is classic.
\end{definition}

\begin{theorem}\label{thm:main}
 For any set of classic $S$--inequations $A$, the category of representations of $(S,A)$ has an initial object.
\end{theorem}

\begin{proof}
   The basic ingredients for building the initial representation 
  are given by the initial representation $\Delta_*\Sigma$ in the category $\Rep^{\Delta}(S)$ (cf.\ Lem.\ \ref{lem:initial_wo_ineq}) 
  or, equivalently, by the initial representation $\Sigma$ in $\Rep(S)$.
  We call $\Sigma$ the monad underlying the representation $\Sigma$.
  
  The proof consists of 3 steps: at first, we define a preorder $\leq_A$ on the terms of $\Sigma$, induced by the set $A$ of inequations.
   Afterwards we show that the data of the representation $\Sigma$ --- substitution, representation morphisms etc. --- 
  is compatible with the preorder $\leq_A$ in a suitable sense. This will yield a representation $\Sigma_A$ of $(S,A)$.
  Finally we show that $\Sigma_A$ is the initial such representation.

\noindent
\emph{--- The monad underlying the initial representation:}

\noindent
    For any set $X$, we equip $\Sigma X$ with a preorder $A$ by setting, for $x,y\in \Sigma X$,

    \begin{equation} x \leq_A y \quad :\Leftrightarrow \quad \forall R : \Rep(S,A), \quad
                      i_R (x) \leq_R i_R (y) \enspace , 
      \label{eq:order}
    \end{equation}
where $i_R : \Sigma \to \bar{R}$ is the initial morphism of representations in monads coming from Zsid\'o's theorem (or, equivalently,
the initial morphism $\Delta_*\Sigma \to R$).
We have to show that the map 
 \[ X\mapsto \Sigma_A X := (\Sigma X, \leq_A) \] 
yields a relative monad over $\Delta$. 
The missing fact to prove is that the substitution of the monad $\Sigma$ with a morphism 
\[ f\in\PO(\Delta X, \Sigma_A Y) \cong \SET(X,\Sigma Y)\] 
is compatible with the order $\leq_A$:
given any $f \in \PO(\Delta X, \Sigma_A Y)$ we show that $\sigma^\Sigma(f) \in \Set(\Sigma X,\Sigma Y)$ is monotone with respect
to $\leq_A$ and hence (the carrier of) a morphism $\sigma(f) \in \PO(\Sigma_A X, \Sigma_A Y)$.
We overload the infix symbol $\bind{}{}$ to denote monadic substitution.
Suppose $x\leq_A y$, we show 
        \[\bind{x}{f} \enspace \leq_A \enspace \bind{y}{f} \enspace .\] 
Using the definition of $\leq_A$, 
we must show, for any representation $R$ of $(S,A)$, 
\[  i_R(\bind{x}{f}) \enspace \leq_R \enspace i_R(\bind{y}{f}) \enspace .\]
Since $i_R$ is a morphism of representations, it is compatible with the substitutions of $\Sigma$ and $\bar{R}$; we have
 \[ i_R(\bind{x}{f}) \enspace = \enspace \bind{i_R(x)} {\comp{f}{i_R}} \enspace . \]
Rewriting this equality and its equivalent for $y$ in the current goal yields the goal 
\[ \bind{i_R(x)} {\comp{f}{i_R}} \quad \leq_A \quad \bind{i_R(y)} {\comp{f}{i_R}} \enspace ,\]
which is true since the substitution of $R$ (whose underlying map is that of $\bar{R}$) is monotone in the first argument 
(cf.\ Rem.\ \ref{rem:about_substitution})
 and $i_R (x) \leq_R i_R(y)$ by assumption.
We hence have defined a monad $\Sigma_A$ on $\Delta$.
We interrupt the proof for an important lemma:
\begin{lemma}\label{lem:useful_lemma}
  Given a classic functor $V : \Rep^{\Delta}(S) \to \LMod(\Delta,\PS)$ from the category of representations in monads on $\Delta$ 
  to the large category of modules over such modules with codomain category $\PS$, we have
   \[ x \leq_A y \in V(\Sigma)(X) \quad \Leftrightarrow \quad \forall R : \Rep(S,A), \quad V(i_R)(x)\leq_{V^R_X} V(i_R)(y) \enspace ,\]
 where now and later we omit the argument $X$, e.g.\ in $V(i_R)(X)(x)$.
\end{lemma}
\begin{proof}[Proof of Lemma \ref{lem:useful_lemma}.]
   The proof is done by induction on the derivation of ``$V$ classic''. The only interesting case is where $V = M\times N$ is a product:
    \begin{align*} 
	  (x_1, y_1) \leq (x_2,y_2) &\Leftrightarrow x_1 \leq x_2 \wedge y_1 \leq y_2 \\
                         {}         &\Leftrightarrow  \forall R, M(i_R) (x_1) \leq M(i_R) (x_2) \wedge \forall R, N (i_R) (y_1) \leq N(i_R) (y_2) \\
                         {}         &\Leftrightarrow  \forall R, M(i_R) (x_1) \leq M(i_R) (x_2) \wedge  N (i_R) (y_1) \leq N(i_R) (y_2) \\
                         {}         &\Leftrightarrow  \forall R, V(i_R) (x_1,y_1) \leq V(i_R) (x_2,y_2) \enspace .
    \end{align*}

\end{proof}

\noindent
\emph{--- Representing $S$ in $\Sigma_A$:}

\noindent
Any arity $s \in S$ should be represented by the module morphism $s^{\Sigma}$, i.e.\ the representation of $s$ in $\Sigma$. 
We have to show that those representations are compatible with the preorder $A$.
Given $x\leq_A y$ in $\dom(s,\Sigma)(X)$, we show (omitting the argument $X$ in $s^{\Sigma}(X)(x)$)
 \[ s^{\Sigma} (x) \quad \leq_A \quad s^{\Sigma}(y) \enspace. \]
By definition, we have to show that, for any representation $R$ as before,
\[ i_R (s^{\Sigma} (x)) \quad \leq_R \quad i_R (s^{\Sigma}(y)) \enspace. \]
Since $i_R$ is a morphism of representations, it commutes with the representational module morphisms ---
the corresponding diagram is similar to the diagram of Def.~\ref{def:morph_of_reps}.
By rewriting with this equality we obtain the goal
\[ s^R \Bigl(\bigl(\dom(s)(i_R)\bigr)(x)\Bigr) \quad \leq_R \quad s^R\Bigl(\bigl(\dom(s) (i_R)\bigr)(y)\Bigr) \enspace. \]
This goal is proved by instantiating Lem.\ \ref{lem:useful_lemma} with the classic $S$--module $\dom(s)$ (cf.\ Def.\ \ref{def:arity_classic_module_untyped})  
and the fact that $s^R$ is monotone.
We hence have established a representation -- which we call $\Sigma_A$ -- of $S$ in $\Sigma_A$.

\noindent
\emph{--- $\Sigma_A$ satisfies $A$:}

\noindent
The next step is to show that the representation $\Sigma_A$ satisfies $A$. 
Given an inequation 
     \[\alpha \leq \gamma : U \to V\] 
of $A$ with a classic $S$--module $V$, 
   we must show that for any set $X$ and any $x\in U(\Sigma_A)(X)$ in the domain of $\alpha$ we have 
       \begin{equation} \alpha^{\Sigma_A}_X(x) \quad \leq_A \quad \gamma^{\Sigma_A}_X(x) \enspace . \label{eq:sigma_a} \end{equation}
In the following we omit the subscript $X$. By Lemma \ref{lem:useful_lemma} the goal is equivalent to

\begin{equation}
 \forall R : \Rep(S,A), \quad V(i_R) (\alpha^{\Sigma_A}(x)) \quad \leq_{V^R_X} \quad V(i_R) (\gamma^{\Sigma_A} (x)) \enspace . \label{eq:sigma_a_alt}
\end{equation}
Let $R$ be a representation of $(S,A)$. We continue by proving \eqref{eq:sigma_a_alt} for $R$.
   By Remark \ref{rem:half_equation_comm} and the fact that $i_R$ is also the carrier of a 
   morphism of $S$--re\-presen\-tations from $\Delta\Sigma$ to $R$ (cf.\ Lemma \ref{lem:adjunction_reps}) we can rewrite the goal as
       \[ \alpha^R\bigl(U(i_R)(x)\bigr) \quad \leq_{V^R_X} \quad \gamma^R \bigr(U(i_R)(x)\bigr) \enspace , \]
  which is true since $R$ satisfies $A$.

\noindent
\emph{--- Initiality of $\Sigma_A$:}

\noindent
Given any representation $R$ of $(S,A)$, the morphism $i_R$ is monotone with respect to the orders on $\Sigma_A$ and $R$ by construction of $\leq_A$.
It is hence a morphism of representations from $\Sigma_A$ to $R$.
Unicity of the morphisms $i_R$ follows from its unicity in the category of representations of $S$, i.e.\ without inequations.
Hence $\Sigma_A$ is the initial object in the category of representations of $(S,A)$.
\end{proof}

\begin{example}[Ex.\ \ref{ex:sig_ulc_prop} continued]
 The only inequation of the signature $\Lambda\beta$ is classic. The initial representation of $\Lambda\beta$
 is given by the monad $\ULCB$ together with the $\ULCB$--module morphisms $\Abs$ and $\App$ (cf.\ Ex.\ \ref{ex:ulcb_constructor_mod_mor})
  as representation structure. 
\end{example}

\subsection{Some Remarks}

\begin{remark}[about ``generating'' inequations]
 Given a 2--signature $(S,A)$ and a representation $R$ of $S$,
  the representation morphism of modules $s^R$ of any arity $s$ of $S$ 
     is monotone.
For the initial representation $\Sigma_A$ of $(S,A)$ this means that any reduction between terms of $\Sigma$, which is specified by
an inequation of $A$,
 is automatically propagated into subterms.
Similarly, the reduction relation on $\Sigma$ generated by $A$ is by construction reflexive and transitive, since we consider representations 
in monads with codomain $\PO$.
For the example of $\Lambda\beta$ this means that in order to obtain the ``complete'' reduction relation $\twoheadrightarrow_{\beta}$, 
it is sufficient to specify only one inequation for the $\beta$--rule 
   \[ \lambda M (N) \enspace \leq \enspace M[*:=N] \enspace . \]
\end{remark}

\begin{remark}[about substitution]\label{rem:about_substitution}
 The substitution in Ex.\ \ref{ex:ulcbeta} is compatible with the order on terms in the following sense:
 \begin{packenum}
  \item $M \leq N \quad$ implies $\quad M[*:=A] \leq N[*:=A] \quad$ and
  \item $A \leq B\quad$ implies $\quad M[*:=A] \leq M[*:=B]$. \label{item:covar_in_subst_arg}
 \end{packenum}
 The first implication is a general fact about relative monads on $\Delta$: 
  for any such monad $P$ and any $f\in \PO(\Delta V,PW)$, the substitution $\sigma_{X,Y}(f) \in \PO(PV,PW)$ 
  is monotone.

The second monotonicity property, however, is not encoded in the framework we give in the 
present paper.
%
%
A different definition of monad which would enforce implication \ref{item:covar_in_subst_arg} to hold 
is given by considering $\PO$ as a category enriched over itself, or as a 2--category: 
given morphisms $f, g : \PO(V,W)$ we say that 
\[ f \Rightarrow g  \quad\text{iff}\quad
 f \leq g \quad\text{iff}\quad \forall v : V, f(v) \leq g(v) \enspace . \]
A monad $P$ would then have to be equipped with a substitution action that is given, for any two sets $V$ and $W$, by \emph{a functor} (of preorders)
\[ \sigma_{V,W} : \PO(\Delta V,PW) \to \PO(PV,PW) \enspace . \]

\noindent
Using this ``enriched'' definition of monad 
--- which is employed in another work of ours \parencite{DBLP:conf/wollic/Ahrens12} ---
for the representations of a 2--signature, we can prove that
any language defined by a 2--signature satisfies the second implication as well.
The proof is available in our \textsf{Coq} library.
\end{remark}

\begin{remark}[about finite contexts]\label{rem:finite_contexts}
 \textcite{DBLP:conf/fossacs/AltenkirchCU10} characterize the untyped lambda calculus as
a relative monad on the inclusion functor $i:\Fin\to\Set$ from finite sets to sets.
 An anonymous referee suggested combining our viewpoint --- syntax as monad over $\Delta:\Set\to\PO$ ---  
with \citeauthor{DBLP:conf/fossacs/AltenkirchCU10}'s: one might consider the 
 lambda calculus as a relative monad on the composition $\comp{i}{\Delta}:\Fin\to\PO$, and, more generally,
 one might consider representations of a signature $(S,A)$ over monads on $\comp{i}{\Delta}:\Fin\to\PO$.
 The above theorem remains true when replacing monads on $\Delta$ by monads on $\comp{i}{\Delta}$
 everywhere.
 An equivalence between the theorem thus obtained and our Thm.\ \ref{thm:main} might be established in a
  way similar to what \citeauthor{ju_phd} does in her PhD thesis:
 she shows, by means of adjunctions between the respective categories of models, the equivalence between the approach
 of \textcite{fpt} --- based on monoids over finite contexts --- and the approach of \textcite{DBLP:conf/wollic/HirschowitzM07},
  where models are built from  monads on the category $\Set$, i.e.\ over arbitrary contexts.
\end{remark}

\begin{remark}[about monads on $\PO$]\label{rem:endo_ord}
 As mentioned in Sec.\ \ref{sec:rel_work}, \textcite{DBLP:journals/njc/GhaniL03} and \textcite{DBLP:journals/iandc/HirschowitzM10}
 suggest the use of monads over the category $\PO$ of preordered sets for modelling syntax with a rewriting relation.
 Indeed, representations of a signature $(S,A)$ could be analogously defined for such monads.
The above construction of the initial representation of $(S,A)$ 
  carries over to representations in such monads,
 thus yielding an initiality result in which syntax is modelled as monad on $\PO$.
 It might be interesting to establish a precise connection --- e.g., in form of adjunctions ---
  between the resulting categories of representations in monads on $\PO$ and representations
  in relative monads on $\Delta$.
\end{remark}

%% file: prop_arities_formal.tex
\section{Formalization in the proof assistant \textsf{Coq}} \label{subsec:1-sig-formal}

In this section we explain some elements of the formalization of the initiality theorem in the proof assistant \textsf{Coq}.
However, we only explain the implementations of definitions and lemmas that are specific to the theorem.
We base ourselves on a general library of category theoretic concepts the formalization details of which we do not go into.
The interested reader can find an in--depth description and the complete \textsf{Coq} code 
online\footnote{\sourceurl}.
The implementation of categories, monads and modules over monads 
(which are analogous to the implementation of their relative counterparts used here)
is explained in detail by \textcite{ahrens_zsido}.

For a morphism $f$ from object $a$ to object $b$ in any category we write \lstinline!f : a ---> b!. 
Composition of morphisms $f : a \to b$ and $g : b \to c$ is written \lstinline!f ;; g!.

\subsection{Arities as Lists}

According to Def.\ \ref{def:1--signature},
a 1--signature consists of an indexing type and, for each index, a list of natural numbers,
indicating the number of arguments of a constructor, as well as the number of variables 
bound in each argument.
In the formalization they are simply called ``signatures'':

 \begin{lstlisting}
Notation "[ T ]" := (list T) (at level 5).
Record Signature : Type := {
  sig_index : Type ;
  sig : sig_index -> [nat] }.
 \end{lstlisting}

\noindent
Next we formalize context extension according to a natural number, cf.\ Sec.\ \ref{subsec:deriv}.
These definitions are important for the definition of the module morphisms we 
associate to an arity, cf.\ below.
Context extension is actually functorial.
Given a natural number \lstinline!n! 
and a set of variables \lstinline!V!, 
 we recursively
define the set \lstinline!V ** n! to be the set \lstinline!V! enriched with \lstinline!n!
additional variables:

 \begin{lstlisting}
Fixpoint pow (n : nat) (V : TYPE) : TYPE :=
  match n with
  | 0 => V
  | S n' => pow n' (option V)
  end.
Notation "V ** n" := (pow n V) (at level 10).
Fixpoint pow_map (l : nat) V W (f : V ---> W) :
         V ** l ---> W ** l :=
  match l return V ** l ---> W ** l with
  | 0 => f
  | S n' => pow_map (^ f)
  end.
Notation "f ^^ l" := (pow_map (l:=l) f) (at level 10).
 \end{lstlisting}

\subsection{Representations}

Given an  arity $s$, i.e.\ a list of natural numbers $s$ (cf.\ Def.\ \ref{def:1--signature}), and a relative monad $P$ on the functor $\Delta$,
we need to define the product module $P^{s}$. More generally, we define $M^s$ for any $P$--module $M$ with codomain $\PO$.
We build this module from scratch instead of 
relying on the category--theoretic constructions such as product and derivation functor for the module categories,
allowing us to omit the insertion of isomorphisms in the style of Lems.\ \ref{pb_prod} and \ref{pb_comm}.
The reasons for this design choice are explained elsewhere \parencite{ahrens_zsido}.
Given any module \lstinline!M! over a monad \lstinline!P! 
from sets to preordered sets, we define the 
product type \lstinline!prod_mod_c! as an inductive type family parametrized by a 
context and dependent on a list of naturals. 
Actually we define at first the carrier depending not on a module, but just on a carrier map \lstinline!M : TYPE -> Ord!.
The relation on the product is induced by that on \lstinline!M!.
\begin{lstlisting}
Variable M : TYPE -> Ord.
Inductive prod_mod_c (V : TYPE) : [nat] -> Type :=
  | TTT : prod_mod_c V nil
  | CONSTR : forall b bs,
         M (V ** b)-> prod_mod_c V bs -> prod_mod_c V (b::bs) .
Notation "a -:- b" := (CONSTR a b) (at level 60).
Inductive prod_mod_c_rel (V : TYPE) : forall n, relation (prod_mod_c M V n) :=
  | TTT_rel : forall x y : prod_mod_c M V nil, prod_mod_c_rel x y
  | CONSTR_rel : forall n l, forall x y : M (V ** n),
          forall a b : prod_mod_c M V l, x << y -> 
     prod_mod_c_rel a b -> prod_mod_c_rel (x -:- a) (y -:- b).
\end{lstlisting}
Note that the infixed ``\lstinline!<<!'' is overloaded notation and denotes the relation of any preordered set.
For any given list \lstinline!a! of naturals and any set \lstinline!V! 
of variables, the set \lstinline!prod_mod_c V a! equipped
with the relation \lstinline!prod_mod_c_rel V a! is in fact a preordered set. For the proof of transitivity we rely
on the \textsf{Coq} tactic \lstinline!dependent induction!, thus on the axioms
\begin{lstlisting}
JMeq.JMeq_eq : forall (A : Type) (x y : A), x ~= y -> x = y 
Eqdep.Eq_rect_eq.eq_rect_eq : forall (U : Type) (p : U) 
                                (Q : U -> Type) (x : Q p) 
                                (h : p = p), x = eq_rect p Q x p h
\end{lstlisting}
from the \textsf{Coq} standard library.
Now if \lstinline!M! is not just a map of type \lstinline!TYPE -> Ord!, 
but a module over some relative monad \lstinline!P! over \lstinline!Delta! (the functor $\Delta:\Set\to\PO$),
we equip the product map with a module substitution in form of a recursive function:

 \begin{lstlisting}
Variable M : RModule P Delta.
Fixpoint pm_mkl l V W (f : Delta V ---> P W)
      (X : prod_mod_c (fun V => M V) V l) : prod_mod_c _ W l :=
     match X in prod_mod_c _ _ l return prod_mod_c (fun V => M V) W l with
     | TTT => TTT _ W
     | elem -:- elems =>
      rmkleisli (RModule_struct := M) (lshift _ f) elem -:- pm_mkl f elems
     end.
(* ... *)
Definition prod_mod (a : [nat]) := Build_RModule (prod_mod_struct a).
 \end{lstlisting}
Afterwards we prove by induction that this map is indeed monotone with respect to the preorder 
\lstinline!prod_mod_c_rel! defined above. Altogether, we obtain 
a module \lstinline!prod_mod M l! for any module \lstinline!M : RMOD P Ord! and
any list of naturals \lstinline!l!.

To any arity \lstinline!ar : [nat]! and a module \lstinline!M! over a monad \lstinline!P!
we associate a type of module morphisms \lstinline!modhom_from_arity ar M!.
Representing \lstinline!ar! in \lstinline!M! then means giving a term of type \lstinline!modhom_from_arity ar M!.
In Def.\ \ref{def:rep_in_rmonad}
we have defined re\-presen\-ta\-tions in \emph{monads} only. Indeed we instantiate \lstinline!M! with 
the tautological module later.

\begin{lstlisting}
Variable P : RMonad Delta.
Definition modhom_from_arity (M : RModule P Ord) (ar : [nat]) : Type := RModule_Hom (prod_mod M ar) M.
\end{lstlisting}

\noindent
For the rest of the section, a representation \lstinline!S! is fixed via a \textsf{Coq} section variable:
\begin{lstlisting}
Variable S : Signature.
\end{lstlisting}

\noindent
As just mentioned, representing the signature \lstinline!S! in a monad \lstinline!P! (cf.\ Def.\ \ref{def:rep_in_rmonad}) means
providing a suitable module morphism for any arity of \lstinline!S!, i.e.\ providing, for any 
element of the indexing set \lstinline!sig_index S!, a term of type \lstinline!modhom_from_arity P (sig i)!:

 \begin{lstlisting}
Definition Repr (P : RMonad Delta) := 
  forall i : sig_index S, modhom_from_arity P (sig i).
Record Representation := {
  rep_monad :> RMonad Delta ;
  repr : Repr rep_monad }.
 \end{lstlisting}
%
Note that the projecton \lstinline!rep_monad! is declared as a \emph{coercion} via the special syntax \lstinline!:>!.
This coercion allows for abuse of notation in \textsf{Coq} as we do informally according to Def.\ \ref{def:rep_in_rmonad},
cf.\ Sec.\ \ref{subsec:ineq_formal} for a use of this coercion.

\subsection{Morphisms of Representations}

A morphism of representations from $P$ to $Q$ ist given by a monad morphism $f : P \to Q$
between the underlying monads such that a diagram commutes for any arity, cf.\ Def.\ \ref{def:morph_of_reps}.
The main task in the implementation is to define this diagram for a given arity $\ell$, and, more
specifically, the left vertical morphism 
\[\dom(\ell,f) = f ^ {\ell} : P^{\ell} \to f^* Q^{\ell} \enspace ,  \] 
using the notation of Def.\ \ref{def:alg_notation}.
Since the carrier of $P^{\ell}$ is defined as an inductive type, it makes sense to define $f^{\ell}$ by recursion on the inductive type underlying $P^{\ell}$,
named \lstinline!prod_mod_c P V l!:
\begin{lstlisting}
Variables P Q : RMonad Delta.
Variable f : RMonad_Hom P Q.
Fixpoint Prod_mor_c (l : [nat]) (V : TYPE) (X : prod_mod_c (fun V => P V) V l) :
                   (prod_mod_c _ V l) :=
  match X in prod_mod_c _ _ l
  return f* (prod_mod Q l) V with
  | TTT => TTT _ _
  | elem -:- elems => f _ elem -:- Prod_mor_c elems
  end.
\end{lstlisting}

\noindent
Proving this map monotone is an easy exercise, as well as its commutation property with substitution,
yielding the aforementioned module morphism.

Now we have all the ingredients we need in order to define the diagram of 
Def.\ \ref{def:morph_of_reps}. For an arity $a$ the diagram reads as follows:

 \begin{lstlisting}
Variable a : [nat].
Variable RepP : modhom_from_arity P a.
Variable RepQ : modhom_from_arity Q a.
Notation "f * M" := (# (PbRMOD f _ ) M).
Definition commute := Prod_mor a ;; f * RepQ == RepP ;; f^.
 \end{lstlisting}
Here \lstinline!f^! denotes the module morphism induced by a monad morphism, cf.\ Def.\ \ref{def:ind_mod_hom}.
Using the preceding definition, we define morphisms of representations of \lstinline!S!:

 \begin{lstlisting}
Variables P Q : Representation.
Class Representation_Hom_struct (f : RMonad_Hom P Q) :=
   repr_hom_s : forall i : sig_index S, commute f (repr P i) (repr Q i).
Record Representation_Hom : Type := {
  repr_hom_c :> RMonad_Hom P Q;
  repr_hom :> Representation_Hom_struct repr_hom_c }.
 \end{lstlisting}

\subsection{Category of Representations}

In this section we describe in more detail the category of representations of a 1--signature, cf.\ Def.\ \ref{lem:cat_of_1-reps}.
The composition of morphisms of representations $f : P \to Q$ and $g : Q \to R$ is essentially done by composing the underlying 
monad morphisms. One has to show that this morphism indeed commutes with the representation morphisms of $P$ and $R$.
Similarly, the identity monad morphism of (the monad underlying) a representation $P$ yields a morphism of representations.
Fed with some suitable lemma, the \lstinline!Program! framework does the job for us:

 \begin{lstlisting}
Program Instance Rep_comp_struct : 
      Representation_Hom_struct (RMonad_comp f g).
Program Instance Rep_Id_struct : Representation_Hom_struct (RMonad_id P).
 \end{lstlisting}

\noindent
Since equality of morphisms of representations is defined as equality of the underlying monad 
morphisms, the categorical properties of composition are established already as part of the 
definition of the category \lstinline!RMONAD F! for any functor \lstinline!F!.

 \begin{lstlisting}
Program Instance REP_struct : Cat_struct (@Representation_Hom S) := {
  mor_oid a c := eq_Rep_oid a c;
  id a := Rep_Id a;
  comp P Q R f g := Rep_Comp f g }.
Definition REP := Build_Cat REP_struct.
 \end{lstlisting}

\subsection{Initiality without Inequations}


We construct the initial object of the category \lstinline!REP!.
In the informal proof of Thm.\ \ref{lem:initial_wo_ineq} this initial object is the image under a left adjoint of the initial object 
in a category of representations as defined by \textcite{DBLP:conf/wollic/HirschowitzM07}.
For the formal proof we decide to implement the initial object of \lstinline!REP! directly, 
in order to obtain a compact formalization.
 \begin{lstlisting}
Inductive UTS (V : TYPE) : TYPE :=
  | Var : V -> UTS V
  | Build : forall (i : sig_index S), UTS_list V (sig i) -> UTS V
with
UTS_list (V : TYPE) : [nat] -> Type :=
  | TT : UTS_list V nil
  | constr : forall b bs,
      UTS (V ** b) -> UTS_list V bs -> UTS_list V (b::bs).
Notation "a -::- b" := (constr a b).
Definition UTS_sm V := Delta (UTS V).
 \end{lstlisting}
%
We define renaming and, built on top of renaming, substitution:

\begin{lstlisting}
Fixpoint rename (V W: TYPE ) (f : V ---> W) (v : UTS V):=
    match v in UTS _ return UTS W with
    | Var v => Var (f v)
    | Build i l => Build (l //-- f)
    end
with
  list_rename V t (l : UTS_list V t) W (f : V ---> W) : UTS_list W t :=
     match l in UTS_list _ t return UTS_list W t with
     | TT => TT W
     | constr b bs elem elems => elem //- f ^^ b -::- elems //-- f
     end
where "x //- f" := (rename f x)
and "x //-- f" := (list_rename x f).
Fixpoint subst (V W : TYPE) (f : V ---> UTS W) (v : UTS V) :
 UTS W := match v in UTS _ return UTS _ with
    | Var v => f v
    | Build i l => Build (l >>== f)
    end
with
 list_subst V W t (l : UTS_list V t) (f : V ---> UTS W) : UTS_list W t :=
     match l in UTS_list _ t return UTS_list W t with
     | TT => TT W
     | elem -::- elems =>
       elem >== _lshift f -::- elems >>== f
     end
where "x >== f" := (subst f x)
and "x >>== f" := (list_subst x f).
\end{lstlisting}

\noindent
Renaming and substitution as just defined correspond to functoriality and
monadic substitution for representations in monads as done by \textcite{DBLP:conf/wollic/HirschowitzM07}.
We, according to Lem.\ \ref{lem:rmon_delta_endomon}, have to apply the functor \lstinline!Delta! 
in some places:
\begin{lstlisting}
Program Instance UTS_sm_rmonad : RMonad_struct Delta UTS_sm := {
  rweta c := #Delta (@Var c);
  rkleisli a b f := #Delta (subst f) }.
Canonical Structure UTSM := Build_RMonad UTS_sm_rmonad.
\end{lstlisting}

\noindent
The monad \lstinline!UTSM! is easily equipped with a representation of the signature \lstinline!S!;
the carrier of the representation of \lstinline!i : sig_index S! is given by 
the function 
\begin{lstlisting}
fun (X : prod_mod_c _ V (sig i)) => Build (i:=i) (UTSl_f_pm (V:=V) X)
\end{lstlisting}
that is, by the constructor \lstinline!Build i! of the type \lstinline!UTS!, precomposed
with an isomorphism \lstinline!UTSl_f_pm! from \lstinline!prod_mod_c UTS! to \lstinline!UTS_list!.
We thus obtain a representation \lstinline!UTSRepr! of the signature \lstinline!S!.

Given another representation, say, \lstinline!R!, of \lstinline!S!, the morphism \lstinline!init!
from \lstinline!UTSRepr! to \lstinline!R!
is defined by recursion:
\begin{lstlisting}
Fixpoint init V (v : UTS V) : R V :=
        match v in UTS _ return R V with
        | Var v => rweta (RMonad_struct := R) V v
        | Build i X => repr R i V (init_list X)
        end 
with
 init_list l (V : TYPE) (s : UTS_list V l) : prod_mod R l V :=
    match s in UTS_list _ l return prod_mod R l V with
    | TT => TTT _ _
    | elem -::- elems => init elem -:- init_list elems
    end.
\end{lstlisting}
This map \lstinline!init! is compatible with lifting and substitution in \lstinline!UTSM! and \lstinline!R!, respectively:
\begin{lstlisting}
Lemma init_lift V x W (f : V ---> W) :
   init (x //- f) = rlift R f (init x).
Lemma init_kleisli V (v : UTS V) W (f : Delta V ---> UTS_sm W) :
  init (v >== f) = rkleisli (f ;; @init_sm W) (init v).
\end{lstlisting}
where \lstinline!init_sm W! is the (trivially) monotone version of \lstinline!init W! --- 
 the adjunct of \lstinline!init W ! under the adjunction of Def.\ \ref{def:delta}.
The latter of those lemmas constitutes an important part of the proof that \lstinline!init! is the carrier of a module 
morphism from \lstinline!UTSM! to \lstinline!R!.
It is trivial to prove that \lstinline!init! is also compatible with the representation structure of \lstinline!UTSRepr! and \lstinline!R!,
thus the carrier of a morphism of representations called \lstinline!init_rep : UTSRepr ---> R!.
Afterwards unicity of \lstinline!init_rep! is proved:

\begin{lstlisting}
Lemma init_unique :forall f : UTSRepr ---> R , f == init_rep.
\end{lstlisting}
\noindent
Finally we establish initiality by an instance declaration of the corresponding class.
The proof field of the class stating unicity is filled automatically by the \lstinline!Program! framework,
using the aforementioned lemma:
\begin{lstlisting}
Program Instance UTS_initial : Initial (REP S) := {
  Init := UTSRepr ;
  InitMor R := init_rep R }.
\end{lstlisting}

\subsection{Inequations and Initial Representation of a 2--Signature}\label{sec:ineq_formal}\label{subsec:ineq_formal}

For a 1--signature $S$, an $S$-module is defined to be a functor from representations of $S$ to the category 
whose objects are pairs of a monad $P$ and a module $M$ over $P$, cf.\ Def.\ \ref{def:cat_halfeq}.
The use of the cumbersome category $\LMod(\Delta,\PS)$ of pairs ensures that representations in a monad $P$
go to $P$--modules. In \textsf{Coq} we represent this dependency via a dependent type.
We do not make use of the functor properties of $S$--modules.

The below definition makes use of two \emph{coercions}. Firstly, we may write $a:\C$ because the 
``object'' projection of the category record (whose definition we omit) is declared as a coercion.
Secondly, the monad underlying any representation can be accessed without explicit projection
using the coercion we mentioned above.

\begin{lstlisting}
Record S_Module := {
  s_mod :> forall R : REP S, RMOD R wOrd ;
  s_mod_hom :> forall (R T : REP S)(f : R ---> T),
         s_mod R ---> PbRMod f (s_mod T)  }.
Notation "U @ f" := (s_mod_hom U f)(at level 4). 
\end{lstlisting}
Note that we write \lstinline!U@f! for the image of the morphism of representations \lstinline!f! under
the $S$--module \lstinline!U!. Source and target module of \lstinline!f! are implicit arguments in this application.

A half-equation is a natural transformation between $S$-modules. We need the naturality condition in the following.
Since we have not formalized $S$-modules as functors, we have to state naturality explicitly:

\begin{lstlisting}
Class half_equation_struct (U V : S_Module) 
    (half_eq : forall R : REP S, U R ---> V R) := {
  comm_eq_s : forall (R T : REP S)  (f : R ---> T), 
      U @ f ;; PbRMod_Hom _ (half_eq T) ==  half_eq R ;; V @ f }.
Record half_equation (U V : S_Module) := {
  half_eq :> forall R : REP S, U R ---> V R ;
  half_eq_s :> half_equation_struct half_eq }.
\end{lstlisting}

\noindent
We now formalize \emph{classic} $S$--modules.
Any list of natural numbers specifies uniquely a classic $S$--module, cf.\ Def.\ \ref{def:alg_s_mod}.
Given a list of naturals \lstinline!codl!, we call this $S$--module \lstinline!S_Mod_classic codl!.
A \emph{classic half--equation} is a half--equation with a classic co\-do\-main, 
and a classic inequation is a pair of parallel classic half--equations
(cf.\ Def.\ \ref{def:ineq_soft}):

\begin{lstlisting}
Definition half_eq_classic (U : S_Module)(codl : [nat]) := 
                               half_equation U (S_Mod_classic codl).
Record ineq_classic := {
  Dom : S_Module ;
  Cod : [nat] ;
  eq1 : half_eq_classic Dom Cod ;
  eq2 : half_eq_classic Dom Cod }.
\end{lstlisting}

\noindent
Give a representation \lstinline!P! and a (classic) inequation \lstinline!e!, 
we check whether \lstinline!P! satisfies  \lstinline!e! by pointwise comparison (cf.\ Def.\ \ref{def:rep_ineq}):

\begin{lstlisting}
Definition satisfies_ineq (e : ineq_classic) (P : REP S) :=
  forall c (x : Dom e P c), eq1 _ _ _ x << eq2 _ _ _ x.
Definition Inequations (A : Type) := A -> ineq_classic.
Definition satisfies_ineqs A (T : Inequations A) (R : REP S) :=
      forall a, satisfies_ineq (T a) R.
\end{lstlisting}

\noindent
We formalize sets of classic inequations as pairs of an indexing type \lstinline!A! together with a term of type \lstinline!Inequations A!,
that is, a map from \lstinline!A! 
to the type of classic inequations \lstinline!ineq_classic!.
The category of representations of $(S,A)$ is obtained as a full subcategory of the category of representations of $S$.
The following declaration produces a subcategory from predicates on the type of representations and 
on the (dependent) type of morphisms of representations, yielding the category \lstinline!PROP_REP! of representations of
$(S,A)$:

\begin{lstlisting}
Variable A : Type.
Variable T : Inequations A.
Program Instance Ineq_Rep : SubCat_compat (REP S)
     (fun P => satisfies_ineqs T P) (fun a b f => True).
Definition INEQ_REP : Cat := SubCat Ineq_Rep.
\end{lstlisting}

\noindent
We now construct the initial object of \lstinline!INEQ_REP!.
The relation on the initial object is defined precisely as in the paper proof, cf.\ Eq.\ \eqref{eq:order}:
\begin{lstlisting}
Definition prop_rel_c X (x y : UTS S X) : Prop :=
      forall R : PROP_REP, init (FINJ _ R) x << init (FINJ _ R) y.
\end{lstlisting}
Here, \lstinline!FINJ _ R! denotes the representation \lstinline!R! as a representation of \lstinline!S!,
i.e.\ the injection of \lstinline!R! in the category \lstinline!REP S! of representations of \lstinline!S!.
The relation defined above is indeed a preorder, and 
we define the monad \lstinline!UTSP! to be the monad whose underlying sets are identical 
to \lstinline!UTSM!, namely the sets defined by \lstinline!UTS!, but equipped with this new preorder.
This monad \lstinline!UTSP! is denoted by $\Sigma_A$ in the paper proof.

The representation module morphisms of the initial representation \lstinline!UTSRepr! can be reused
after having proved their compatibility with the new order, yielding a re\-pre\-sen\-ta\-tion \lstinline!UTSProp!.
This representation satisfies the inequations of \lstinline!T!:
\begin{lstlisting}
Lemma UTSPRepr_sig_prop : satisfies_ineqs T UTSProp.
\end{lstlisting}
We explicitly inject the representation into the category of representations of $(S,A)$:
\begin{lstlisting}
Definition UTSPROP : INEQ_REP :=
 exist (fun R : Representation S => satisfies_ineqs T R) UTSProp
  UTSPRepr_sig_prop.
\end{lstlisting}

\noindent
In order to build the initial morphism towards any representation \lstinline!R : INEQ_REP!, 
we first build the corresponding morphism in the category of representations of $S$:
\begin{lstlisting}
Definition init_prop_re : UTSPropr ---> (FINJ _ R) := ...
\end{lstlisting}
which we then inject, analoguously to the initial representation, into the subcategory of 
representations of $(S,A)$:
\begin{lstlisting}
Definition init_prop : UTSPROP ---> R := exist _ (init_prop_re R) I.
\end{lstlisting}

\noindent
Finally we obtain our Thm.\ \ref{thm:main}. An initial object of a category is given by
an object \lstinline!Init! of this category, a map associating go any object \lstinline!R!
a morphism \lstinline!InitMor R : Init ---> R!, and a proof of uniqueness of any such morphism.
We instantiate the type class \lstinline!Initial! for the category \lstinline!INEQ_REP!
of representations of $(S,A)$:

\begin{lstlisting}
Program Instance INITIAL_INEQ_REP : Initial INEQ_REP := {
  Init := UTSPROP ;
  InitMor := init_prop ;
  InitMorUnique := init_prop_unique }.
\end{lstlisting}

\noindent
We check its type after closing all the sections, thus abstracting from the section variables:
\begin{lstlisting}
Check INITIAL_INEQ_REP.
INITIAL_INEQ_REP
     : forall (S : Signature) (A : Type) (T : Inequations S A),
       Initial (INEQ_REP (S:=S) (A:=A) T)
\end{lstlisting}

\subsection{\texorpdfstring{$\Lambda\beta$}{Lambdabeta}: Lambda Calculus with \texorpdfstring{$\beta$}{Beta} reduction}

We implement the example 2--signature $\Lambda\beta$, cf.\ Ex.\ \ref{ex:sig_ulc_prop}.
Throughout this section, we use use a custom notation in \textsf{Coq} for the datatype of lists:
\begin{lstlisting}
Notation "[[ x ; .. ; y ]]" := (cons x .. (cons y nil) ..). 
\end{lstlisting}

\subsubsection{The 1--Signature \texorpdfstring{$\Lambda$}{Lambda}}

In order to specify the 1--signature $\Lambda$ (cf.\ Ex.\ \ref{ex:sig_ulc}), 
we first 
define an indexing set \lstinline!Lambda_index! consisting of two elements, \lstinline!ABS! and \lstinline!APP!.
This indexing set reflects the fact that the signature $\Lambda$ consists of two arities.
The record instance \lstinline!Lambda! is a term of type \lstinline!Signature!.
The map \lstinline!sig Lambda! then associates the corresponding lists of naturals to each of these elements, according to 
Def.\ \ref{def:1--signature}:
 
\begin{lstlisting}
Inductive Lambda_index := ABS | APP.
Definition Lambda : Signature := {|
  sig_index := Lambda_index ;
  sig := fun x => match x with 
                  | ABS => [[ 1 ]] 
                  | APP => [[ 0 ; 0]]
                  end |}.
\end{lstlisting}

\subsubsection{The \texorpdfstring{$\Lambda$}{Lambda}--Inequation \texorpdfstring{$\beta$}{beta}}

The definition of the inequation $\beta$ (cf.\ Ex.\ \ref{ex:sig_ulc_prop}) is a more challenging task, since a half--equation is not 
just combinatory data like a 1--arity, but given by suitable module morphisms.
At first, we define the substitution of \emph{one} variable (cf.\ Def.\ \ref{def:subst_half_eq}) as a half--equation.
The carrier \lstinline!subst_carrier! of the substitution is defined as in Def.\ \ref{ex:subst}.
Afterwards we prove that this carrier satisfies the properties of a module morphism, that is, is compatible with 
substitution in the source and target modules. After abstracting from the section variable \lstinline!R!,  
we obtain a function \lstinline!subst_module_mor!
which, given any representation \lstinline!R! of \lstinline!S!, yields the
substitution module morphism associated to (the monad underlying) \lstinline!R!.

\begin{lstlisting}
Variable S : Signature.
Variable R : REP S.
Definition subst_carrier :
   (forall c : TYPE, (S_Mod_classic_ob [[1; 0]] R) c ---> 
         (S_Mod_classic_ob [[0]] R) c) := ...
Program Instance sub_struct : RModule_Hom_struct
      (M:=S_Mod_classic_ob [[1 ; 0]] R) 
      (N:=S_Mod_classic_ob [[0]] R) 
   subst_carrier.
Definition subst_module_mor := Build_RModule_Hom (sub_struct R).
\end{lstlisting}

\noindent
The last step is to prove ``naturality''.
We recall that we do not implement $S$--modules as functors, but just as the data part of functors. 
This is why we put the word \emph{naturality} in quotes. After the proof we define our first half--equation, \lstinline!subst_half_eq!.

\begin{lstlisting}
Program Instance subst_half_s : half_equation_struct 
      (U:= S_Mod_classic [[1 ; 0]]) 
      (V:= S_Mod_classic [[0]]) 
   subst_module_mor.
Definition subst_half_eq := Build_half_equation subst_half_s.
\end{lstlisting}

\noindent
The definition of the second half--equation of Ex.\ \ref{ex:app_circ_half} is possible for any 1--signature with abstraction and application, such as 
the 1--signature $\Lambda$. To keep the example simple, we only define the half--equation for $\Lambda$.
The needed steps are precisely the same as for the substitution half--equation, so we just give the statements.
\begin{lstlisting}
Definition beta_carrier :
   (forall c : TYPE, (S_Mod_classic_ob [[1; 0]] R) c ---> 
         (S_Mod_classic_ob [[0]] R) c) := ...
Program Instance beta_struct : RModule_Hom_struct
      (M:=S_Mod_classic_ob [[1 ; 0]] R) 
      (N:=S_Mod_classic_ob [[0]] R) 
   beta_carrier.
Definition beta_module_mor := Build_RModule_Hom beta_struct.
Program Instance beta_half_s : half_equation_struct
      (U:=S_Mod_classic Lambda [[1 ; 0]])
      (V:=S_Mod_classic Lambda [[0]])
   beta_module_mor.
Definition beta_half_eq := Build_half_equation beta_half_s.
\end{lstlisting}

\noindent
We package both half--equations into one inequation, the beta rule of Ex.\ \ref{ex:sig_ulc_prop}:

\begin{lstlisting}
Definition beta_rule : ineq_classic Lambda := {|
   eq1 := beta_half_eq ;
   eq2 := subst_half_eq Lambda |}.
\end{lstlisting}

\noindent
We can now associate a short name to the category of representations of $\Lambda\beta$, where, for increased clarity,
we specify the implicit arguments:
\begin{lstlisting}
Definition Lambda_beta_Cat := INEQ_REP
    (S:=Lambda)(A:=unit)(fun x : unit => beta_rule).
\end{lstlisting}
Our formal definition allows that an inequation appears multiple times in a 2--signature, whereas in the informal definition we
have \emph{sets} of inequations.
Unlike for arities, having several copies of the same inequation 
does not change the resulting category or its initial object.
The initial representation is obtained via the specification
\begin{lstlisting}
Definition Lambda_beta := @Init _ _ _ 
        (INITIAL_INEQ_REP (fun x : unit => beta_rule)).
\end{lstlisting}